\documentclass[10pt,leqno]{article}

\usepackage[a4paper,tmargin=2cm,bmargin=2.2cm,rmargin=2.5cm,lmargin=2.5cm]{geometry}

\usepackage{amsmath,amsthm,amssymb}
\usepackage{mathtools}
\usepackage{centernot}
\usepackage{subcaption}
\usepackage[normalem]{ulem}
\usepackage{tikz}
\usepackage[utf8]{inputenc}
\usepackage[T1]{fontenc}
\usepackage{verbatim}
\usepackage{float}
\usepackage{authblk}
\usepackage{hyperref}
\usepackage{graphicx}

\newtheorem{Thm}{Theorem}
\newtheorem{Def}{Definition}

\newcommand{\K}{\mathcal{K}}
\newcommand{\M}{\mathcal{M}}
\newcommand{\sR}{\mathbb{R}}
\newcommand{\T}{\mathcal{T}}
\newcommand{\N}{\mathcal{N}}

	\newcommand{\Cc}{\mathcal{C}}
	\newcommand{\C}{\mathbb{C}}

	\DeclareMathOperator{\supp}{supp}

    \newcommand{\vc}{\vcentcolon =}

	\newcommand{\Dslash}{{\mathcal{D} \mkern-11.5mu/\,}} 
	\newcommand{\Pf}{\mathfrak{P}}

	\newcounter{mnotecount}[section]
	\renewcommand{\themnotecount}{\thesection.\arabic{mnotecount}}
	\newcommand{\mnote}[1]
	{\protect{\stepcounter{mnotecount}}$^{\mbox{\footnotesize
				$
				\bullet$\themnotecount}}$ \marginpar{
			\raggedright\tiny\em
			$\!\!\!\!\!\!\,\bullet$\themnotecount: #1} }

	\definecolor{darkgreen}{rgb}{0,.5,0}

\usepackage{bm}
\newcommand{\bx}{\boldsymbol{x}}
\newcommand{\by}{\boldsymbol{y}}

\title{
Quantifying superluminal signalling in Schr\"odinger--Newton model
}

\author[1,2]{Julia Os\k{e}ka-Lenart
}
\author[3,4]{Marcin P{\l}odzie\'n
}
\author[4,5]{Maciej Lewenstein}
\author[6]{Micha{\l} Eckstein
}
\affil[1]{Astronomical Observatory, Faculty of Physics, Astronomy and Applied Computer Science Jagiellonian University, Orla 171, Krak\'ow, 30-244, Ma{\l}opolska, Poland}

\affil[2]{Doctoral School of Exact and Natural Sciences, Jagiellonian University, ul.\ {\L}ojasiewicza 11, 30–348 Krak\'ow, Poland}

\affil[3]{Qilimanjaro Quantum Tech, Carrer de Vene\k{c}uela 74, 08019 Barcelona, Spain}

\affil[4]{ICFO-Institut de Ciencies Fotoniques, The Barcelona Institute of Science and Technology, 08860 Castelldefels (Barcelona), Spain}

\affil[5]{ICREA, Passeig Lluis Companys 23, 08010 Barcelona, Spain}

\affil[6]{Institute of Theoretical Physics, Faculty of Physics, Astronomy and Applied Computer Science, Jagiellonian University, ul.\ {\L}ojasiewicza 11, 30–348 Krak\'ow, Poland}

\begin{document}

\maketitle

\begin{abstract}
The Schr\"odinger--Newton equation aims at describing the dynamics of massive quantum systems subject to the gravitational self-interaction. As a deterministic nonlinear quantum wave equation, it is generally believed to conflict with the relativistic no-signalling principle. Here we challenge this viewpoint and show that it is of key importance to study the quantitative and operational character of the superluminal effects. To this end we employ a rigorous formalism of probability measures on spacetime and quantify the probability of a successful superluminal bit transfer via the single-particle Schr\"odinger--Newton equation. We demonstrate that such a quantity decreases with the increasing size and mass of the system. Furthermore, we prove that the Einstein--Dirac system, which yields the Schr\"odinger--Newton equation in the non-relativistic limit, is perfectly compatible with the relativistic causal structure. Our study demonstrates that the Schr\"odinger--Newton equation, which is by construction non-relativistic, is in fact `more compatible' with the no-signalling principle than the ordinary free  Schr\"odinger equation.
\end{abstract}

\section{Introduction \& motivation}
		
		One of the most outstanding challenges in modern physics is to explain how the classical macro-world emerges from the inherently quantum micro-world. Are the laws of quantum mechanics universally valid or do they break down at some intermediate scale? This question has been studied intensively over the past years both from theoretical and empirical side \cite{OR,QG_RMP,Zurek_decoh}.
		
		A hallmark quantum feature, which is readily absent in the classical world, is the possibility for physical objects to be in spatial superposition. A possible solution to this conundrum is provided by the various \emph{objective collapse models} \cite{OR}. Their common denominator is the assumption that there exists a physical mechanism, which induces an effective localisation of the wave packets of sufficiently massive objects. These models are subject to increasing experimental scrutiny \cite{Superpos2019,Diosi21,collapse22}, with no evidence so far in favour of any of them.
		
		An exciting perspective opens up if one conjectures that gravity is somehow involved in the wavefunction collapse mechanism \cite{QG_RMP}. One of the early elegant models was developed by Di\'osi \cite{Diosi87} and Penrose \cite{Penrose96}. It is based on the so-called \textit{Schr\"odinger--Newton} equation, which describes a self-gravitating massive quantum particle \cite{SN_review_14}. It can be derived from a system of coupled Schr\"odinger--Poisson equations \cite{Penrose96} or semi-classical gravity \cite{Diosi87} or as a Hartee approximation to the full quantum gravity treatment \cite{SN_review_14}.

                Other nonlinear modifications of the Schr\"odinger equation, not directly related to gravity, have also been proposed in the literature \cite{IBB_nonlinear,WeinbergNQM,WeinbergNQM2,Czachor1998,CzachorMarcin,Czachor02,Kent05,Helou17,Caban20,Caban21,Kaplan2022,Paterek24,NJP2025}. All these fall into the class of \emph{deterministic} nonlinear quantum wave equations.
		
		There is, however, a major problem associated with deterministic nonlinear quantum dynamics: It seems to be irreconcilable with relativity. More precisely, Gisin has proven \cite{Gisin1989,Gisin2001} that nonlinear evolution of quantum states facilitates, via entanglement and projective measurements, an immediate information transfer between spacelike separated parties (see, however, \cite{Czachor1998,CzachorMarcin,Czachor02,Kent05,Helou17,NJP2025,Caban20,Caban21,Kaplan2022,Paterek24} for some models evading this problem).
		
		On the other hand, the standard \emph{linear} quantum dynamics also seems to facilitate superluminal signalling. Indeed, Hegerfeldt has shown \cite{Hegerfeldt1} that a strictly localised (i.e. compactly supported in space) initial wave packet immediately develops infinite tails under the Schr\"odinger evolution with a positive-definite Hamiltonian (see, however, \cite{Yngvason,QIandGR} for a critical interpretation). This result was then extended to Gaussian wave packets \cite{Hegerfeldt1985} and sharpened in \cite{PRA2017} basing on the formalism of Borel probability measures on spacetime \cite{AHP2017}. Finally, in \cite{PRA2020} it was shown that the Schr\"odinger equation typically does facilitate operational superluminal signalling via local detection, with the notable exception of the Dirac and photon equations \cite{PRA2017}.
		
		While this result sounds rather disturbing if one scrupulously applies the ``no-signalling principle'' (see e.g. \cite{Bell_Nonlocal,PR_box,PawelRaviCausality}), it has to be interpreted from a physical perspective. After all, it is not surprising that the \emph{nonrelativistic} quantum wave packet formalism is not compatible with the light-cone structure of the relativistic spacetime. The vital question is not \emph{if} causality is violated in a given model, but \emph{how much} it is violated. The tools developed in \cite{PRA2017,AHP2017,PRA2020} (see also \cite{WSWSG11}) allow for a rigorous formulation of this problem in terms of characteristic space- and time-scales and the probability of a successful superluminal bit transfer. Such an approach allows one to evaluate the adequacy of wavepacket evolution models, basing on a quantitative estimate of their incompatibility with the light-cone structure.

		\medskip
        
        The central goal of this paper is to assess and quantify the superluminal effects in the Schr\"odinger--Newton equation. This model is covered neither by Gisin's argument \cite{Gisin1989,Gisin2001}, which concerns multi-particle equations, nor by Hegerfeldt's theorem, which does not apply to nonlinear equations. 
        In order to gain a quantitative insight into the problem we employ the formalism of probability measures on spacetime \cite{AHP2017,PRA2017} and study the spreading of wavepackets. Our main result is that the nonlinear self-interaction term in the Schr\"odinger--Newton equation actually \emph{improves} the relativistic causality properties as compared to the free Schr\"odinger evolution.
        
        The paper is organised as follows:
        
        In Sec. \ref{sec:causal_viol_quant} 
we introduce the necessary elements of the formalism of causality theory for probability measures on spacetime from \cite{AHP2017,PRA2017}. It includes, in particular, a measure of causality violation along with its operational interpretation in terms of the probability of a successful superluminal bit transfer.

    Then, in Sec. \ref{sec:self_grav}, we discuss two dynamical models of self-gravitating quantum systems: In Sec. \ref{sec:S-N_intr} we summarise the Schr\"odinger--Newton equation and in Sec. \ref{sec:E-D} we discuss the Einstein--Dirac system, of which the Schr\"odinger--Newton equation is a rigorous non-relativistic approximation \cite{Giulini12}. We prove that the Einstein--Dirac system always induces a subluminal evolution of probabilities, and hence is perfectly compatible with relativistic causality.

    Sec. \ref{sec:S-N} contains a summary of our numerical studies. We employ a $(1+1)$-dimensional analogue of the Schr\"odinger--Newton equation, which facilitates the numerical simulations without additional symmetry assumptions, while keeping the essential features of the full $(3+1)$-dimensional model. In Sec. \ref{sec:Gauss} we investigate the spreading of an initial Gaussian wavepacket, which allows us to compare and contrast the outcomes with the results from  \cite{PRA2017} for the free nonrelativistic Schr\"odinger equation. We show that for small values of the self-coupling constant the gravitational attraction leads to localisation of the wavepacket and improved causality properties. However, for larger values of the gravitational self-coupling this picture breaks down because part of the mass is ejected to infinity with a superluminal velocity. 

    Then, in Sec. \ref{sec:trap}, we discuss a more sophisticated scenario with a direct operational interpretation: We introduce a box trapping potential, which can be released `at will' by an Alice. We compute the ground state of the system, study its spreading in time and calculate the probability that a Bob, located outside of Alice's future light cone, will detect the released system. We show that the amount of probability which propagates superluminally decreases monotonically with the increasing size and mass of the system.

    We conclude in Sec. \ref{sec:sum} with a summary and consequences of the obtained results. Therein, we also discuss the prospects of simulating the Schr\"odinger--Newton equation in Bose--Einstein condensate. Finally, we make an outlook into the possible future direction of studies and argue that the rejection of the Schr\"odinger--Newton equation on the basis of its alleged conflict with relativity is premature.

		\section{Relativistic causality for nonlocal phenomena}
		\label{sec:causal_viol_quant}
		
		We shall work in the generally covariant relativistic framework. In this setting, a \emph{spacetime} $\M$ is a 4-dimensional time-oriented manifold equipped with a Lorentzian metric (see e.g. \cite{Wald,Beem}). The \emph{causal structure} of $\M$ is a relation between the pairs of point (aka ``events'') in $\M$. Concretely, an event $q$ is said to be in the causal future of an event $p$, denoted $p \preceq q$ or $q \in J^+(p)$, if there exists a piecewise-smooth future-directed causal curve from $p$ to $q$ or $p=q$.
		
		Relativistic spacetimes can be classified into a `ladder-like' structure depending on their causality properties (see  \cite{MingRev,MS08} for the rudiments of causality theory). An important intermediate rung is populated by causal spacetimes, in which there are no closed causal loops, that is, equivalently, the relation $\preceq$ is a partial order. On top of the ladder one finds \emph{globally hyperbolic spacetimes}, which admit a (non-unique) smooth 3+1 splitting $\M \cong \mathbb{R} \times \Sigma$, with a Cauchy hypersurface $\Sigma$. This feature is necessary for the well-posedness of dynamical evolution problems involving Einstein equations coupled to matter \cite{Ringstrom}.
		
		The standard causality theory has been developed for point-like physical objects (`particles') and has a clear interpretation in the spirit of classical physics \cite{Ehlers}. However, in quantum theory the physical systems are typically delocalised in spacetime \cite{QIandGR} and the concept of relativistic causality becomes more obscure (see e.g. \cite{PawelRaviCausality}). Typically, in the context of quantum nonlocality \cite{Bell_Nonlocal}, one adopts the \emph{no-signalling principle}, which says there cannot be any \emph{operational} information transfer between spacelike-separated events in spacetime. In its essence, it means that while physics can be nonlocal it cannot facilitate superluminal communication through any physical devices, not even statistically. This is a rather natural demand, which prevents grandfather-type paradoxes. At the same time, as it turned out, the no-signalling constraints allow for `even-more-nonlocal' probabilistic theories beyond quantum mechanics \cite{PopescuRohrlich94,NPA08,AlmostQ,AQ_hierarchy,PawelRaviCausality,PR_box}. This is the context, in which Gisin's theorem \cite{Gisin1989,Gisin2001} was formulated, casting doubts over the consistency of the nonlinear deterministic modifications of quantum state dynamics.

		\subsection{Causality theory for probability measures}
		
		In \cite{AHP2017,Miller17a,Miller21} a mathematical formalism was tailored to rigorous studies of relativistic causality in nonlocal theories. It was then applied in \cite{PRA2017,PRA2020,JGP2021} to the problem of superluminal wavepacket spreading (cf. \cite{Gerlach1968,Gerlach1969,Gromes1970,Hegerfeldt1,Hegerfeldt1985,WSWSG11}). It has two primary assets: general covariance and operational quantification of causality violation. The former allows one to take into account gravity, also dynamically coupled to matter --- as in the Schr\"odinger--Newton equation \eqref{SN}. The latter allows to quantify the superluminal signalling in terms of characteristic time and length scales, as well as the probability of a successful bit transfer.
		
		Within the formalism \cite{AHP2017} the relevant probabilities are modelled through 
        probability measures (called ``measures'' for short from now on) over a spacetime manifold $\M$. Given such a measure $\mu$ and a region of spacetime $\K \subset \M$ the number $\mu(\K) \in [0,1]$ designates the probability that the signal (which can be carried e.g. by a particle or a field, depending on the considered theory) is registered with a suitable detector operating in $\K$.
		
		On the mathematical side, the space $\Pf(\M)$ of all measures on $\M$ is a convex set, the extreme points of which are the Dirac deltas at the points of $\M$. It can also be seen as the space of all mixed states on the $C^*$-algebra of observables $\Cc(\M,\C)$. The pure states are in 1:1 correspondence with the points in $\M$.
		
		In \cite{AHP2017} it was shown (see also \cite{Suhr16}) that the causal relation $\preceq$ on $\M$ naturally extends to $\Pf(\M)$ via the Lorentzian optimal transport:
		\begin{Def}[\cite{AHP2017}]\label{def:caus}
			We say that a measure $\nu$ on spacetime $\M$ is in the causal future of a measure $\mu$, denoted by $\mu \preceq \nu$, if there exists a joint measure $\omega \in \Pf(\M \times \M)$ such that
			\begin{enumerate}
				\item $\omega(A \times \M) = \mu(A)$ and $\omega(\M \times A) = \nu(A)$ for any Borel $A \subset \M$,
				\item $\omega(J^+) = 1$,
			\end{enumerate}
			where $J^+ = \{ (p,q) \in \M \times \M \, | \, p \preceq q \}$.
		\end{Def}
		
		The above definition encodes the following physical intuition:
		The existence of a joint probability measure $\omega$, with $\mu$ and $\nu$ as its marginals, provides a (non-unique) probability flow from $\mu$ to $\nu$. The condition $\omega(J^+) = 1$ says that the flow is conducted exclusively along future-directed causal curves. 
		
		If the spacetime $\M$ has sufficiently nice causal properties (concretely, if it is causally simple), then Definition \ref{def:caus} has a convenient equivalent formulation.
		
		\begin{Thm}[\cite{AHP2017}]\label{thm:caus}
			Let $\M$ be causally simple. We have, $\mu \preceq \nu$ if and only if
			\begin{align}\label{causK}
				\mu(\K) \leq \nu(J^+(\K)) \quad \text{ for every compact } \quad \K \subset \M.
			\end{align}
		\end{Thm}
		
		The formulation \eqref{causK} says that the probability must not leak out of the causal cone of any compact region of spacetime --- see Fig. \ref{fig:rel}.

		\subsection{Causal evolution of measures}
		
		Let us now assume that $\M$ is globally hyperbolic. We choose a Cauchy temporal function $\T: \M \to \sR$ and consider an \emph{evolution of measures} --- a one-parameter family $\{\mu_t\}_{t \in \sR}$ of measures supported on time-slices, $\supp \mu_t \subset \Sigma_t = \T^{-1}(t)$. Employing Definition \ref{def:caus} one arrives at a natural probabilistic generalisation of a future-directed causal curve.

		\begin{Def}[\cite{PRA2017}]\label{def:CE}
			An evolution of measures $\{\mu_t\}_t$ is \emph{causal} if
			\begin{align}\label{CE}
				\text{for all } \;  s \leq t, \qquad \mu_s \preceq \mu_t.
			\end{align}
		\end{Def}

        A family $\{\mu_t\}_{t \in \sR}$ models an evolution of detection statistics associated with some nonlocal physical phenomenon. More precisely, $\mu_t(\{t\}\times K)$ gives the probability that a detector occupying a volume $K$ would click at a time instant $t$. For instance, if $\psi(t,x)$ is a wave-function describing the evolution of a quantum system, then the position measurement yields 
        \begin{align}
            \mu_t(\{t\}\times K) = \int_K dx \, \vert \psi(t,x) \vert^2.
        \end{align}

        From the relativistic point of view, the time-evolution is always associated with a chosen slicing of spacetime manifold determined by a time function. A natural question is what happens to the property \eqref{CE}, if we choose another time function. It has been proven in \cite{Miller17a,Miller21} that this property is in fact \emph{slicing-independent}. More precisely, given a causal evolution of measures we can always find a geometrical object $\sigma$ (a measure on the space of future-directed causal worldlines), which is manifestly invariant. This means that the description of the same nonlocal physical phenomenon with the help of any other time function $\T'$ will also be causal.

        \subsection{An operational protocol for superluminal signalling}\label{sec:protocol}

        The relationship between the concept of a causal evolution of measures and the no-signalling principle has been analysed in detail in \cite{PRA2020}. In particular, it has been shown that a violation of the property \eqref{CE} always facilates an operational protocol for statistical superluminal bit transfer, under a rather mild assumption about the impact of a local measurement on the system at hand.

        Here we provide an alternative scenario, which allows to exploit the failure of \eqref{CE} for superluminal signalling: Suppose that Alice can locally trap a physical system in a stationary state and that she can also freely choose to release it. The evolution of the trapped system is trivial, i.e. $\mu_t = \mu_0$ for all $t \geq 0$. On the other hand, the trap release in general will induce a growing in time probability to find the system outside of the trap. More formally, let us denote by $m_K$ the Alice's decision of releasing ($m_K = 0$) or keeping ($m_K = 0$) the trap at a time instant $t=0$. The evolution of the untrapped system is modelled by a family of measures $\mu_t(\,\cdot\, | \,  m_K = 0)$ for $t\geq 0$. For the trapped system we clearly have
        \begin{align}
            \mu_t(\,\cdot\, | \,  m_K = 1) = \mu_0(\,\cdot\, | \,  m_K = 1) =  \mu_0(\,\cdot\, | \,  m_K = 0), \quad \text{ for all } \quad t > 0.
        \end{align}

        Suppose now that Bob has a device tuned to detect the system prepared by Alice in a region of space $C$. Then, Bob can infer the bit $m_K$ encoded by Alice from the detection statistics
         \begin{align}
            \mu_t \big( \{t\} \times C \, \big\vert \,  m_K = 1 \big) \neq \mu_t \big( \{t\} \times C \,  \big\vert \,  m_K = 0 \big).
        \end{align}
        Technically, it would require the preparation by Alice of an ensemble of same systems, simultaneously trapped or released, followed by Bob's gathering of information from multiple detectors.

        If the spacetime region $\{t\} \times C$ is \emph{not} in the causal future of the spacetime region $\K \vc \{0\} \times K$, then the information transfer from Alice to Bob is superluminal.

        Suppose now that the evolution of the measure $\mu_t(\,\cdot\, | \,  m_K = 1)$ is \emph{not} causal in the sense of Definition \ref{def:CE}. Then, on the strength of Theorem \ref{thm:caus}, there exists a compact region of spacetime $\Cc \subset \{t\} \times \Sigma_t$ outside of the causal future of $\K$, i.e. $\Cc \subset \M \setminus J^+(\K)$, such that\footnote{Technically, this follows from the fact any measure on spacetime manifold is tight \cite{AHP2017}.}
         \begin{align}
            \mu_t \big( \Cc \, \big\vert \,  m_K = 1 \big) \neq \mu_0 \big( \K \,\big\vert \,  m_K = 1 \big).
        \end{align}
        Consequently, if Bob fills the region $C$, corresponding to $\Cc = \{t\} \times C$, with detectors, he will be able to statistically decode Alice's message. See Figure \ref{fig:rel} for an illustration.

        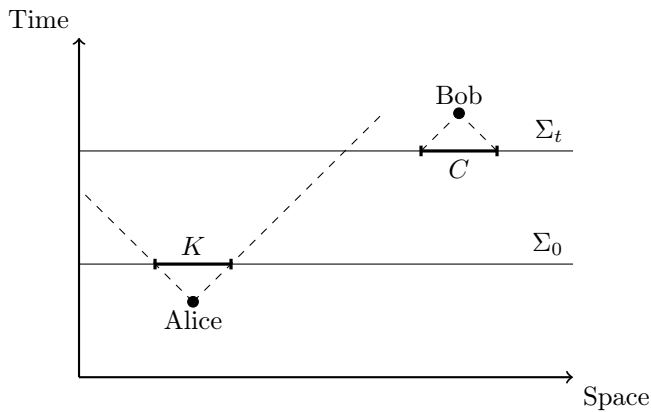
\begin{figure}[H]
    \centering
    \begin{tikzpicture}
    \draw[thick,->] (0,0) -- (6.5,0) node[anchor=north west] {Space};
    \draw[thick,->] (0,0) -- (0,4.5) node[anchor=south east] {Time};
    \draw (0,1.5) -- (6.5,1.5) node[anchor=south east] {$\Sigma_0$};
    \draw (0,3) -- (6.5,3) node[anchor=south east] {$\Sigma_t$};
    \draw[very thick] (1,1.5 cm +2 pt) -- (1,1.5 cm-2 pt);
    \draw[very thick] (2,1.5 cm +2 pt) -- (2,1.5 cm-2 pt);
    \draw[very thick] (1,1.5) -- (2,1.5) 
    ;
    \filldraw[black] (1.5,1.5) circle (0pt) node[anchor=south]{$K$};
    \filldraw[black] (1.5,1) circle (2pt) node[anchor=north]{Alice};
    \draw[dashed] (1.5,1) -- (0,2.5);
    \draw[dashed] (1.5,1) -- (4,3.5);
    \draw[very thick] (4.5,3 cm +2 pt) -- (4.5,3 cm-2 pt);
    \draw[very thick] (5.5,3 cm +2 pt) -- (5.5,3 cm-2 pt);
    \draw[very thick] (4.5,3) -- (5.5,3);
    \filldraw[black] (5,3.5) circle (2pt) node[anchor=south]{Bob};
    \draw[dashed] (4.5,3) -- (5,3.5);
    \draw[dashed] (5.5,3) -- (5,3.5);
    \filldraw[black] (5,3) circle (0pt) node[anchor=north]{$C$};
\end{tikzpicture}
    \caption{Visualization of the superluminal signalling protocol exploiting the violation of condition \eqref{CE}. Alice prepares a system trapped in region $K$ and decides whether to release it or not. Bob gathers information from the detection statistic in region $C$. If condition \eqref{CE} is violated, then Bob can infer (with some probability) Alice's decision.}
    \label{fig:rel}
\end{figure}

        Suppose that the physical system in consideration is described by a wavefunction, as we shall assume from Section \ref{sec:E-D} on. Then, this scenario falls in into the class of communication protocols, where information is encoded by a local change of dynamics, rather than a projective measurement (see \cite{NJP2025}). Let us also observe that we do \emph{not} require the initial wavepacket $\psi_0$, which determines the measure $\mu_0$, to be compactly supported. What counts is the \emph{difference} between Bob's probabilities to detect the particle if the trap is on or off. This sets our work outside of the context of Hegerfeldt's theorem, which has been criticised precisely on the ground of impossibility to operationally prepare a compactly supported wavepacket \cite{Yngvason} (see also \cite{QIandGR,Yngvason2015}).

        \subsection{Quantification of causality violation}

        One of the advantages of the measure-theoretic formalism developed in \cite{AHP2017} is that it allows for quantification of relativistic causality violation in a given physical model. A natural measure associated with the characterisation \eqref{causK}, put forward in \cite{PRA2017}, is the following
        \begin{align}\label{M}
            M(t,K) \vc \max \big\{ 0, \mu_0(\{0\} \times K) - \mu_t \big( J^+(\{0\} \times K) \big) \big\}.
        \end{align}
        The number $M(t,K)$ is the amount of probability found on the time-slice $\Sigma_t$, that leaked out of the future light-cone of the region $K$. If the evolution $\{\mu_t\}_{t \geq 0}$ is causal in the sense of Definition \ref{def:CE} then $M(t,K)=0$ for all $t$ and $K$. In the context of the protocol described in Sec. \ref{sec:protocol} the number $M(t,K) \in [0,1]$ is the maximal success rate of a single bit  superluminal transfer from Alice to Bob, with Alice operating a trap in region $K$.

\section{Self-gravitating quantum systems}\label{sec:self_grav}

\subsection{Schr\"odinger--Newton equation}\label{sec:S-N_intr}

         The so-called \textit{Schr\"odinger--Newton} equation describes a self-gravitating massive quantum particle \cite{SN_review_14}. It arises from a coupled system of Schr\"odinger and Poisson equations \cite{Penrose96}:
		\begin{align}
			i \,\hbar \, \partial_t \psi(t,\bx)& = \Big(-\frac{\hbar^2}{2m}\nabla^2 + V(t,\bx) + m +\Phi(t,\bx)\Big) \psi(t,\bx), \\ \nabla^2 \Phi(t,\bx)  & = 4\pi G m |\psi(t,\bx)|^2, 
		\end{align}
		where $\psi$  denotes the wave function, $t$ is time, $\bx \in \mathbb{R}^3$ is the position, $m$ is the mass of the system, $G$ is the gravitational constant and $V$ is an external potential. By solving the Poisson equation one arrives at the following nonlinear quantum wave equation
		\begin{align}\label{SN}
			i\, \hbar \, \partial_t \psi(t,\bx) = \Big(-\frac{\hbar^2}{2m}\nabla^2+V(t,\bx)-Gm^2\int\frac{|\psi(t,\by)|^2}{|\bx-\by|}d^3\by \Big)\psi(t,\bx).
		\end{align}
        
        Alternatively, Eq. \eqref{SN} can be derived from the semi-classical approach to quantum field theory on curved spacetimes \cite{Diosi87}. Both these derivations are based on the assumption that gravity remains \emph{fundamentally} classical, while all matter fields are quantised \cite{Diosi87,Penrose96,Penrose98,Penrose2014}. A different route to the Schr\"odinger--Newton equation leads through a mean-field Hartree approximation to the (canonically) quantised gravity \cite{SN_review_14}. Consequently, the physical interpretation of Eq. \eqref{SN} depends on the context, in which it was derived. 
        
        It is also important to stress that Eq. \eqref{SN} is a \emph{deterministic} quantum wave equation in the sense that an initial state $\psi(0,\bx) \in L^2(\sR^3)$ determines the evolution of the wavefunction for all $t\geq 0$. This should be distinguished from a \emph{stochastic} version of the Schr\"odinger--Newton equation \cite{Diosi89} called in the literature the Di\'osi--Penrose model \cite{OR,Diosi21,collapse22}.

        The Eq. \eqref{SN} has been thoroughly analysed in the literature. The stationary states were investigated in some detail in \cite{Moroz98,Bernstein98,Tod2001,Harrison2001}. The studies of the Schr\"odinger--Newton dynamics \cite{Harrison2003,Salzman2005,Carlip2008,Giulini2011,Giulini2013,Meter2011,Manfredi2013} revealed that, above a certain mass threshold, the wavepackets tend to localise (see also \cite{DunajskiPenrose23}). The latter feature will prove to be of key importance in our studies of superluminal spreading in Sec. \ref{sec:S-N}.

        The Schr\"odinger--Newton equation appears also in different physical contexts, as an effective equation. It is invoked in some cold dark matter models \cite{Schr_Poiss_Vlas_Poiss,2024PhRvR_VP_SP,2025arXiv_fluid_analog_SN} and in the description of boson stars evolution \cite{boson_stars_NatCom16}. These exotic compact objects can be an alternative to what is currently observed and interpreted as black holes. 
        
		\subsection{Einstein--Dirac system}
\label{sec:E-D}

In \cite{Giulini12} it was shown that the Schr\"odinger--Newton equation \eqref{SN} can be derived as the nonrelativistic limit of the Einstein--Dirac system. More precisely, this result was derived under the assumption of spherical symmetry of the spacetime metric. The coupled Einstein--Dirac equations read:
\begin{align}\label{ED}
 & G_{\mu \nu} = \frac{8\pi \, G}{c^4}\, T_{\mu \nu}(\psi), \\
    & \Big(\Dslash-\frac{mc}{\hbar}\Big)\psi = 0, \label{Dirac}
\end{align}
where $G_{\mu \nu}$ is the Einstein tensor, $\Dslash$ is the Dirac operator associated with the spin connection \cite{Finster98,Finster99} and $T_{\mu \nu}(\psi)$ denotes the stress--energy tensor for the Dirac field $\psi$ \cite{Giulini12}.

In \cite{PRA2017} it has been shown that the free Dirac equation \eqref{Dirac} in Minkowski spacetime induces a causal evolution of measures in the sense of Definition \ref{def:CE}, for any initial spinor. This does not contradict Hegerfeldt's theorem \cite{Hegerfeldt1}, because the Dirac Hamiltonian (in any local frame) is not bounded from below \cite{Thaller}. This result is somewhat surprising (cf. \cite{WSWSG11}) as it shows that the existence of negative-energy solutions, which is commonly taken as a major drawback of one-particle Dirac's theory, does not spoil the relativistic causal character of the evolution. Furthermore, this result still holds in presence of an external electromagnetic or Yang--Mills potential \cite{PRA2017}.

We now show that the system \eqref{ED} with suitable initial conditions also induces a causal evolution of measures in the precise sense of Definition \ref{def:CE}. Let us emphasise that this is a very significant improvement of the theorem derived in \cite{PRA2017}. Firstly, this is so, because it involves the Dirac equation in a curved spacetime and the proof bases on a recent powerful result by Miller \cite{Miller21}, which tights the causal evolution of measures with a generally covariant continuity equation. Secondly, the spacetime metric in the Einstein--Dirac model is no longer a background for the evolution of probabilities, but it couples dynamically to the wavefunction via Einstein equations \eqref{ED}. Nevertheless, the recent results on the existence of solutions to the Einstein--Dirac system \cite{EinsteinDirac25} enable us to set the evolution problem rigorously.

\begin{Thm}\label{thm:ED}
    The Einstein--Dirac equations induce a causal evolution of measures for any sensible initial data and any space--time splitting.
\end{Thm}
\begin{proof}
The solutions of Einstein--Dirac equations which we are interested in consist of a global hyperbolic spacetime $\M$ with a metric tensor $g^{\mu \nu}$ and global spinor field $\psi\in L^2(\M,S(\M))$. In \cite{EinsteinDirac25} it was proven that solutions to the system \eqref{ED} exist semi-globally for initial data specified on two null hypersurfaces $\N$ and $\N'$ that intersect at a spacelike 2-sphere. 
Let us then fix a solution, $(\M,g^{\mu\nu},\psi)$ to the Einstein--Dirac system with some suitable initial data.

The Dirac operator on curved spacetime is defined as
\begin{equation*}
    \Dslash \vc i \tilde{\gamma}^{\mu}\nabla^s_{\mu}= i \tilde{\gamma}^\mu(I\partial_\mu+\Gamma_\mu),
\end{equation*}
where $\nabla^s_{\mu}=(I\partial_\mu+\Gamma_\mu)$ is a spinorial covariant derivative, $\Gamma_\mu$ is the spinor affine connection and $I$ is the identity matrix of dimension 4. The curved gamma matrices $\tilde{\gamma}^\mu$ satisfy the relation
\begin{align*}
\{\tilde{\gamma}^\mu,\tilde{\gamma}^\nu\} = 2 g^{\mu\nu}.
\end{align*}
They are related to the flat ones $\gamma^a$ as follows:
\begin{equation*}
    \tilde{\gamma}^\mu = e^{\mu}_a\gamma^a,
\end{equation*}
where $e^{\mu}_a$ is a tetrad vector set, which provides a transformation to locally flat spacetime: $\eta_{ab}=e^{\mu}_ae^{\nu}_bg_{\mu\nu}$.

It is straightforward to show (see e.g. \cite{DiracCurvedBook}) that because $\psi$ satisfies the Dirac equation (\ref{Dirac}) the 4-vector defined as $j^\mu \vc \psi^\dagger \gamma^0 \tilde{\gamma}^\mu \psi$ 
is covariantly conserved
    \begin{align}
        \nabla_\mu^sj^\mu = 0.
    \end{align}
Because $j^{\mu}$ is a 4-vector and it has no spinor indices, the above equation  implies that $j^\mu$ satisfies the continuity equation:
\begin{equation}
    \nabla_{\mu}j^\mu=0.
    \label{eq:continuity}
\end{equation}

Now, let us pick a Cauchy temporal function $\T: \M \to \sR$, which defines a family of Cauchy hypersurfaces $\Sigma_t = \T^{-1}(t)$, for $t \in \sR$. We denote by $\psi(t,\cdot)$ the projection of a global spinor onto $\Sigma_t$. It determines a one-parameter family of measures 
\begin{align}\label{eq:evo}
d\mu_t(\bx)  = j^0(t,\bx) dS_t(\bx) = \psi^\dagger(t,\bx) \gamma^0 \tilde{\gamma}^\mu(t,\bx) \psi(t,\bx) dS_t(\bx),
\end{align}
where $dS_t$ is the volume element on $\Sigma_t$.

We will show that the evolution of measures \eqref{eq:evo} is causal in the sense of Definition \ref{def:CE}. To this end, we employ Theorem 6 from \cite{Miller21}, which says that an evolution of measures is causal if and only if it satisfies a suitable continuity equation with a subluminal velocity field. This is indeed the case for the family \eqref{eq:evo} induced by the Einstein--Dirac system on the strength of Eq. \eqref{eq:continuity}. 

We thus only need to show that $j^\mu$ is a causal vector field. To prove this let us assume that $j^\mu$ is spacelike at some point $p \in \M$. In such a case, one can find a local inertial frame, where $j'^0(p)=\psi'^\dagger(p)\gamma^0\tilde{\gamma}'^0(p)\psi'(p)=0$. The fact that $\gamma^0\neq 0$ and $\tilde{\gamma}'^0\neq 0$ leads to $\psi'^\dagger(p)\psi'(p)=0$ and consequently $\psi'(p)=0$. Since $\psi'(p)$ can be obtained by local Lorentz transformation of $\psi(p)$, this means that $\psi(p)=0$. But then $j^\mu(p)=0$, which is in contradiction to the initial assumption about the spacelike nature of $j^{\mu}(p)$. Hence, $j^\mu$ is indeed a causal vector field.
\end{proof}

		\section{Schr\"odinger--Newton wavepacket spreading}
		\label{sec:S-N}

        In the previous section we have shown that the fully relativistic Einstein--Dirac system always induces a causal evolution of probabilities. We now turn to its non-relativistic limit --- the Schr\"odinger--Newton system \eqref{SN}. As explained in the Introduction, one should not expect this non-relativistic equation to respect the relativistic causal structure. However, it makes sense to quantify the departure from the relativistic principles using the formalism from Sec. \ref{sec:causal_viol_quant}. 
        
        To this end, we first need to introduce suitable dimensionless units. Let us define the characteristic length $\ell_0$ and time $t_0$ scales. The relation between them is $t_0 = m\,\ell_0^2/\hbar$ and results from the requirement for kinetic and Hamiltonian terms to have the same dimension. This leaves us with one free parameter --- the dimensionless self-coupling constant $\kappa$, which is related to the physical parameters as
			\begin{align}
				\kappa \,=\, Gm^2 \frac{t_0}{\hbar^2 \,\ell_0}\,=\,G\frac{m^3 \ell_0}{\hbar^2}.
			\end{align}
 
The gravitational self-coupling is expected to play a role when $\kappa$ is on the order of one. By setting the mass  $m=Nm_u$ in atomic mass units (Daltons $1\,\mathrm{Da}=m_u$), 
the dimensioneless coupling constant reads $\kappa \approx 2\times 10^{-36} N$.
            Consequently, $\kappa\approx 1$ requires $N$ to be of order $10^{12}$. The current experimental limits \cite{Superpos2019} have reached $N \sim 2.5\cdot10^4$, which would give a negligible self-interaction $\kappa \sim 10^{-23}$ in the Schr\"odinger--Newton model.

        In our numerical analysis we focus on the $(1+1)$-dimensional analogue of Eq. \eqref{SN},
             \begin{equation}
			i \frac{\partial \psi(t,x)}{\partial t} = -\frac{1}{2} \frac{\partial^2\psi(t,x)}{\partial x^2}  - \kappa \int d y \frac{|\psi(t,y)|^2}{|x-y|}\, \psi(t,x).
			\label{eq:SN_dim}
		\end{equation}
        As we shall argue in the concluding Section \ref{sec:sum}, it grasps the essential features of the full $(3+1)$-dimensions yielding a faithful quantitative picture of superluminal signalling, while remaining numerically tractable.

        Our analysis consists in numerically solving Eq. \eqref{eq:SN_dim} for a given initial state $\psi(t=0,x)$, and computing the resulting probability density $\rho(t,x) \vc |\psi(t,x)|^2$. 
 For real- and imaginary-time propagation, split-step Crank--Nicolson discretisations remain a robust choice for trapped condensates with contact or dipolar interactions~\cite{William_H_Press2007-rp,Muruganandam2009,Kumar2015Dipolar}. Ground-state computations are efficiently performed using normalised gradient-flow methods---equivalent to an implicit Euler or Crank--Nicolson step combined with spectral or high-order finite-difference spatial discretisations---which offer strong stability and accuracy guarantees~\cite{BaoDu2004}. Time-dependent simulations employ semi-implicit Crank--Nicolson or high-order Runge--Kutta integrators within finite-difference or pseudospectral frameworks, as summarised in standard reviews of numerical methods for the nonlinear Schr\"odinger dynamics~\cite{AntoineBaoBesse2013}. For the nonlocal interaction term, typically expressed as a convolution, we use fast Fourier transform based solvers that provide high accuracy even for singular or anisotropic kernels~\cite{BaoJiangTangZhang2015,BaoTangZhang2016Dipolar}. These approaches are implemented in widely used simulation packages, such as GPELab~\cite{AntoineDuboscq2014GPELab}. Recently, nonlinear Schr\"odinger solver based on tensor networks has been proposed in \cite{boucomas2025quanticstensortrainsolving}. In this work we use split-step method both for the ground state preparation with imaginary time evolution, as well for real time evolution.

       The wave-function $\psi$ evolving according to the dynamical equation \eqref{eq:SN_dim} defines a family of measures $d\mu_t = \rho(t,x) \, dx$, for $t \geq 0$. To quantify the amount of superluminal signalling we employ the function \eqref{M}, which measures the time-dependent leakout of probability from the future light-cone of a given region of space. Concretely, we have
       	\begin{equation}
		M(t,R,\kappa) = \int_{-R}^{R} \rho(0,x)\, dx - \int_{-R-ct}^{R+ct} \rho(t,x)\, dx
        \label{M_int}
	\end{equation}
	and we set $c=1$. The one-dimensional region $K$ is simply an interval $[-R,R]$, which has a different interpretation depending on the choice of the initial state. We shall also keep the explicit dependence on the coupling $\kappa$, which quantifies the strength of the self-coupling. Function \eqref{M_int} depends also implicitly on the choice of the initial state $\psi(0,x)$ for Eq. \eqref{eq:SN_dim}. In the next subsections we consider two classes of initial states: a Gaussian and the ground state of Eq. \eqref{eq:SN_dim} with a trapping potential.

    It is also useful to calculate the maximal value of $M(t,R,\kappa)$ that occurs during the time evolution for fixed $\kappa$ and $R$, 
        \begin{equation}
            \widetilde{M}(\kappa,R)= \max_{t \in [0,t_\text{max}]} M(t,R,\kappa).
            \label{M_max}
        \end{equation}
It will allow us to compare how the amount of superluminal signalling changes with $R$ and $\kappa$. Typically, we run the simulations for $t_\text{max} = 30$. 

        \subsection{Gaussian wave packet}\label{sec:Gauss}
 
        Let us consider an initial Gaussian wavepacket:
        \begin{equation}\label{Gauss}
        \psi(0,x)=\big(2\pi \sigma^2\big)^{-1/4}e^{-x^2/(4\sigma^2)}
        \end{equation}
        and fix the width $\sigma=1$ for concreteness. In fact, as noted in \cite{Manfredi2013}, for Eq. \eqref{eq:SN_dim} with a Gaussian initial wavepacket the rescaling of the width $\sigma \rightsquigarrow \lambda \sigma$ is equivalent to the rescaling of the coupling $\kappa \rightsquigarrow \lambda \kappa$.

        In Figs. \ref{fig:rho_gauss} and \ref{fig:rho_w_gauss} we present the results of the simulations for $\kappa=0.1$, $\kappa=1$ and $\kappa=3$. One can see immediately that for larger values of $\kappa$ the gravitational self-attraction balances the wavepacket spreading and the system tends to localise. This is consistent with the results for the original ($3+1$)-dimensional equation \eqref{SN} found in the literature \cite{Harrison2003,Salzman2005,Carlip2008,Giulini2011,Giulini2013,Meter2011,Manfredi2013}. The competition between the spreading induced by the free Hamiltonian term and the self-interaction induces a specific oscillation in the probability density, the frequency of which increases with $\kappa$.

        \begin{figure}[H]
		    \centering
		    \begin{subfigure}{0.32\textwidth}
		    \includegraphics[width=\textwidth]{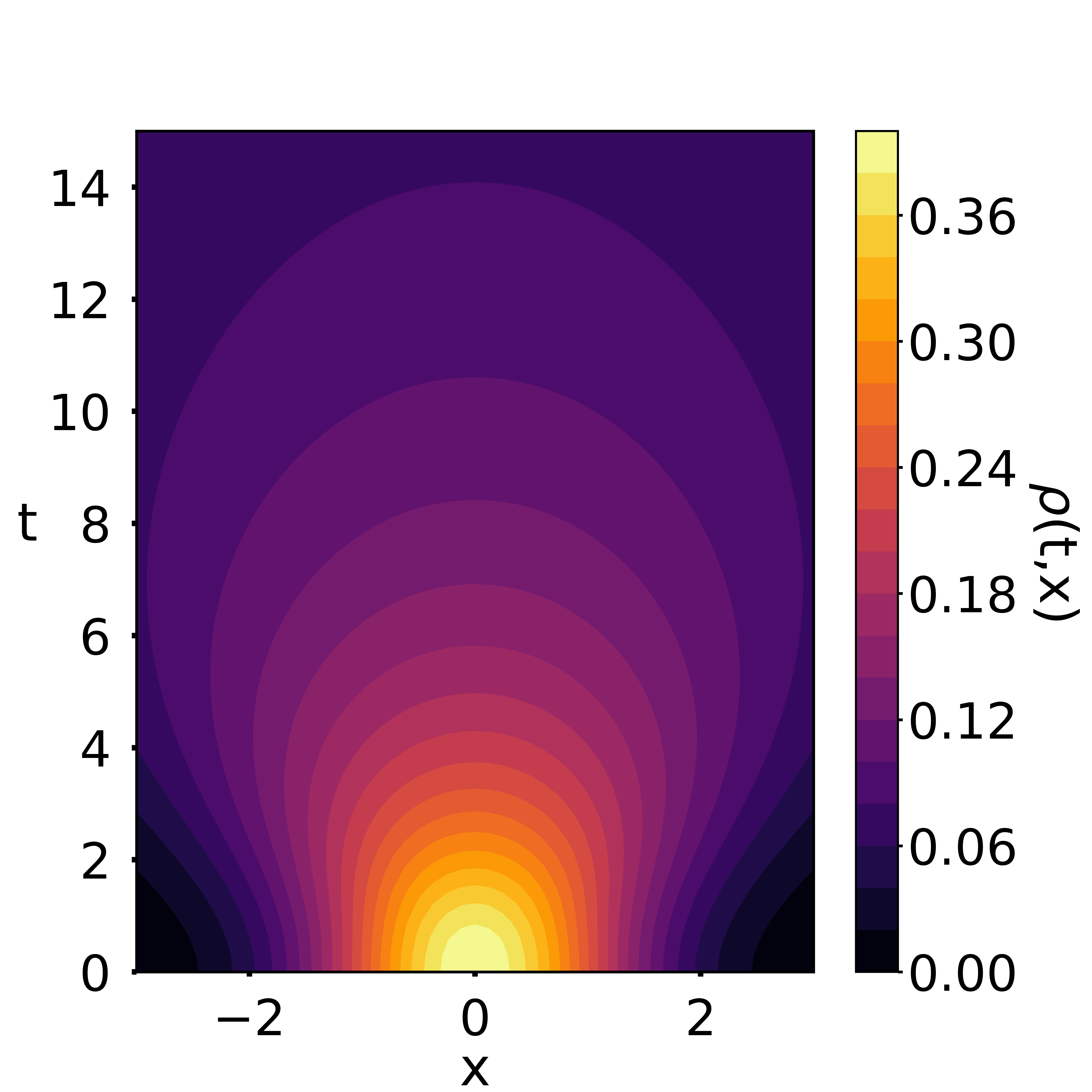}
		        \caption{$\kappa=0.1$}
		        \label{fig:rho_01_gauss}
		    \end{subfigure}
		    \hfill
		    \begin{subfigure}{0.32\textwidth}
		    \includegraphics[width=\textwidth]{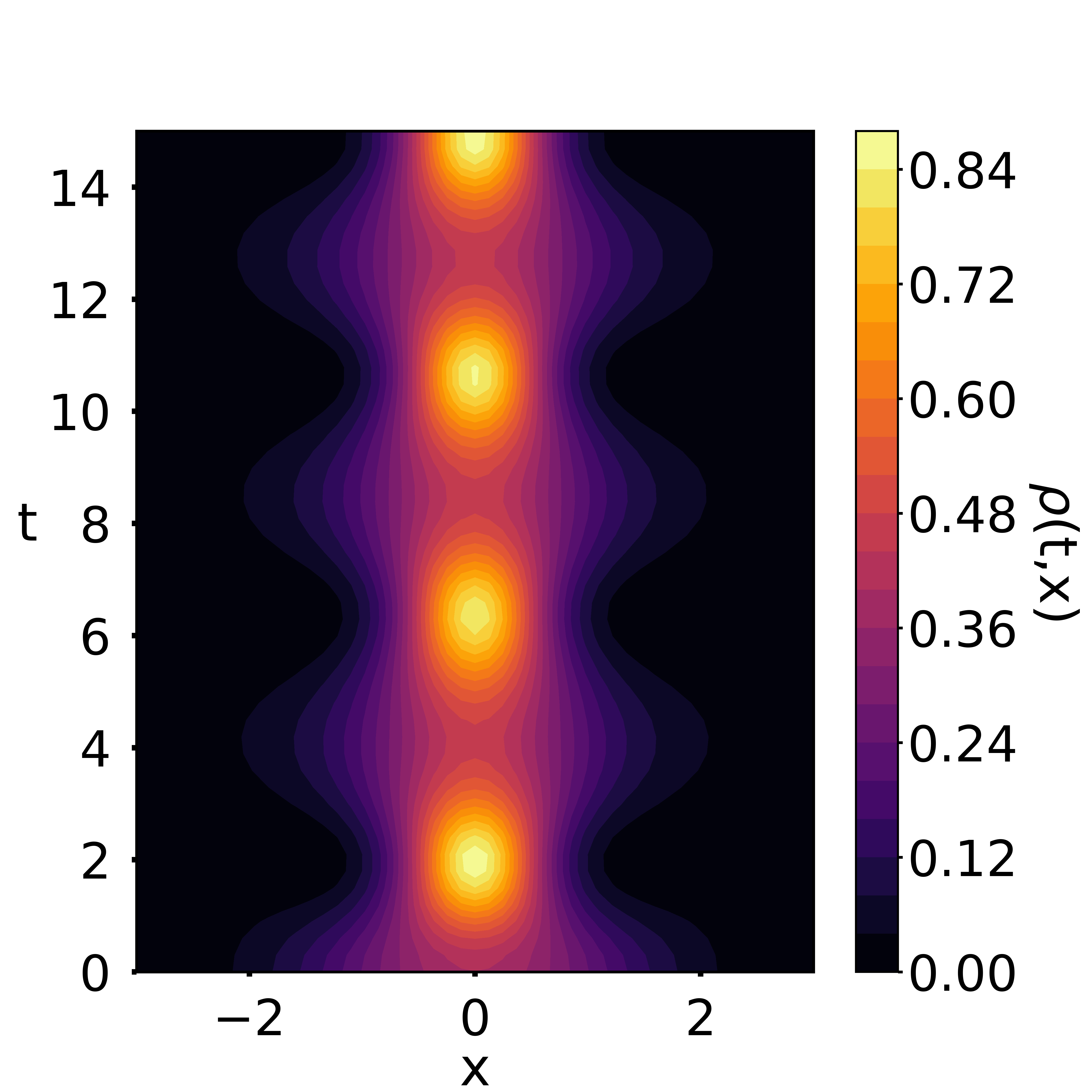}
		        \caption{$\kappa=1$}
		        \label{fig:rho_1_gauss}
		    \end{subfigure}
            \hfill
		    \begin{subfigure}{0.32\textwidth}
		      \includegraphics[width=\textwidth]{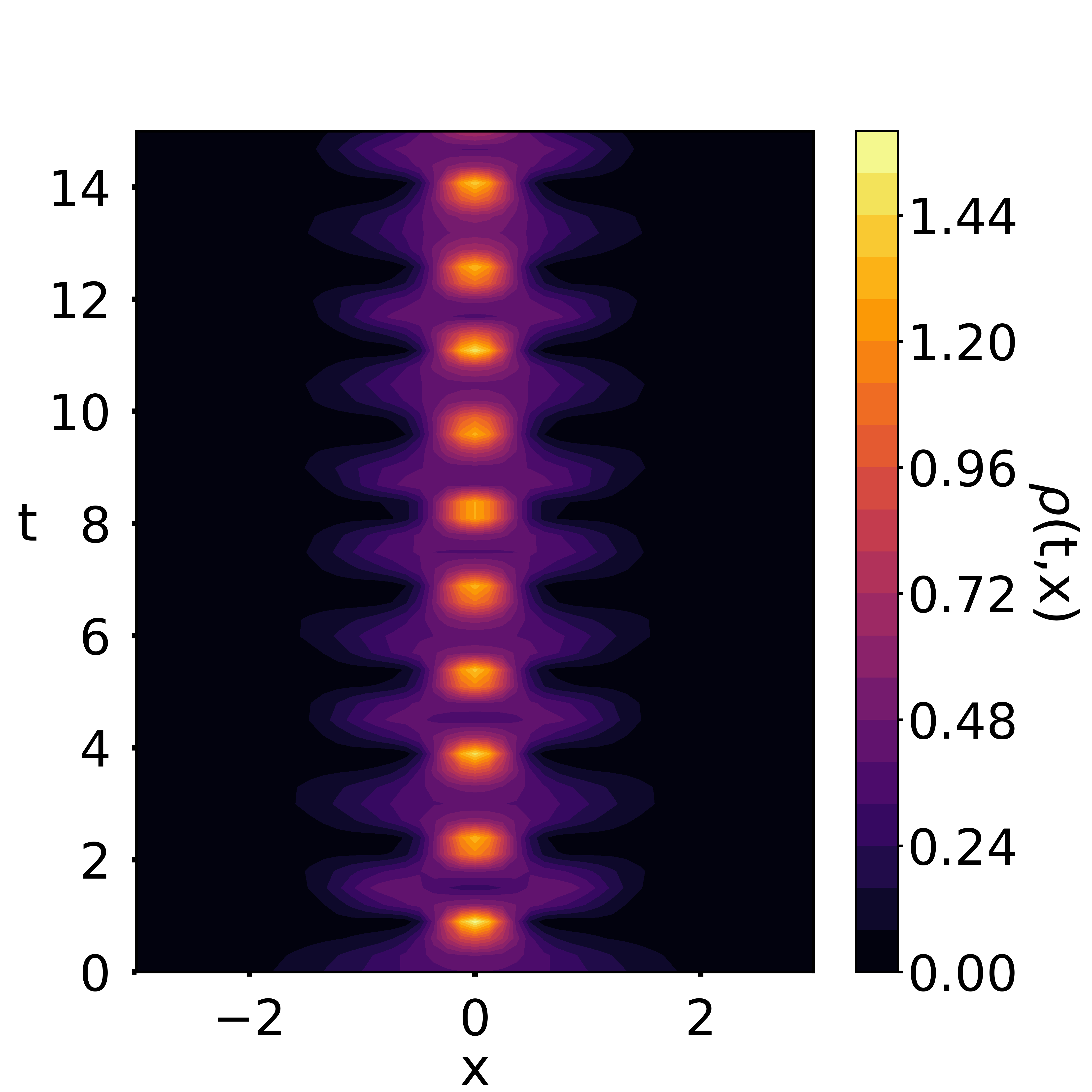}
		        \caption{$\kappa=3$}
		        \label{rho_3_gauss}
		    \end{subfigure}
		    \caption{Time evolution of probability density $\rho (t,x)$ for the Gaussian initial state.}
		    \label{fig:rho_gauss}
		\end{figure}

        There is a critical value of the coupling, $\kappa_c$, for which the self-interaction takes over and the packet tends to self-localise. In Fig. \ref{fig:kappa_c} we show the wavepacket spread,
        \begin{align}\label{spread}
            \Delta_x (t,\kappa) = \sqrt{ \int dx \, x^2\, \vert\psi(t,x)\vert^2 - \left( \int dx \, x\, \vert \psi(t,x) \vert^2 \right)^2 },
        \end{align}
        as a function of time and $\kappa$. It shows that $\kappa_c \approx 0.4$ marks the boundary between the spreading and localisation. 

        \begin{figure}[H]
		    \centering
		    \includegraphics[width=0.4\textwidth]{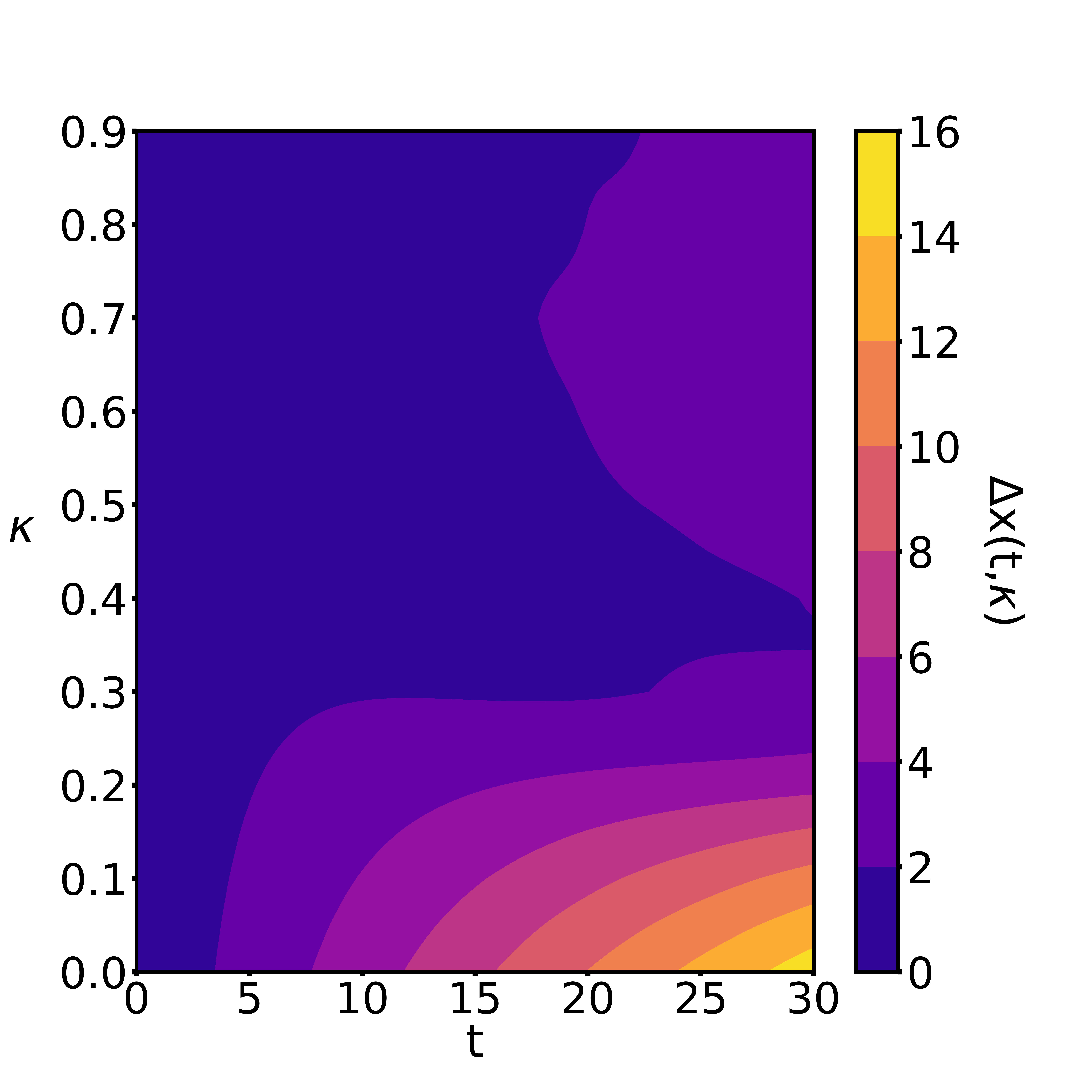}
		    \caption{The wavepacket spread \eqref{spread} for the initial Gaussian state in function of time and the self-coupling strength $\kappa$. For $\kappa_c \gtrsim 0.4$ the wavepacket tends to localise.}
		    \label{fig:kappa_c}
		\end{figure}

       A different insight into the Schr\"odinger--Newton dynamics is offered if we plot the evolution of probability density in logarithmic scale. Fig. \ref{fig:rho_w_gauss} provides a wider perspective on the density plots from Fig. \ref{fig:rho_gauss}. It unveils a significant amount of probability density escaping the light cone. Observe that for $\kappa = 1$ there seems to be less probability beyond the light cone than for $\kappa=0.1$, yet also less than for $\kappa=3$.

        \begin{figure}[H]
		    \centering
		    \begin{subfigure}{0.32\textwidth}
		    \includegraphics[width=\textwidth]{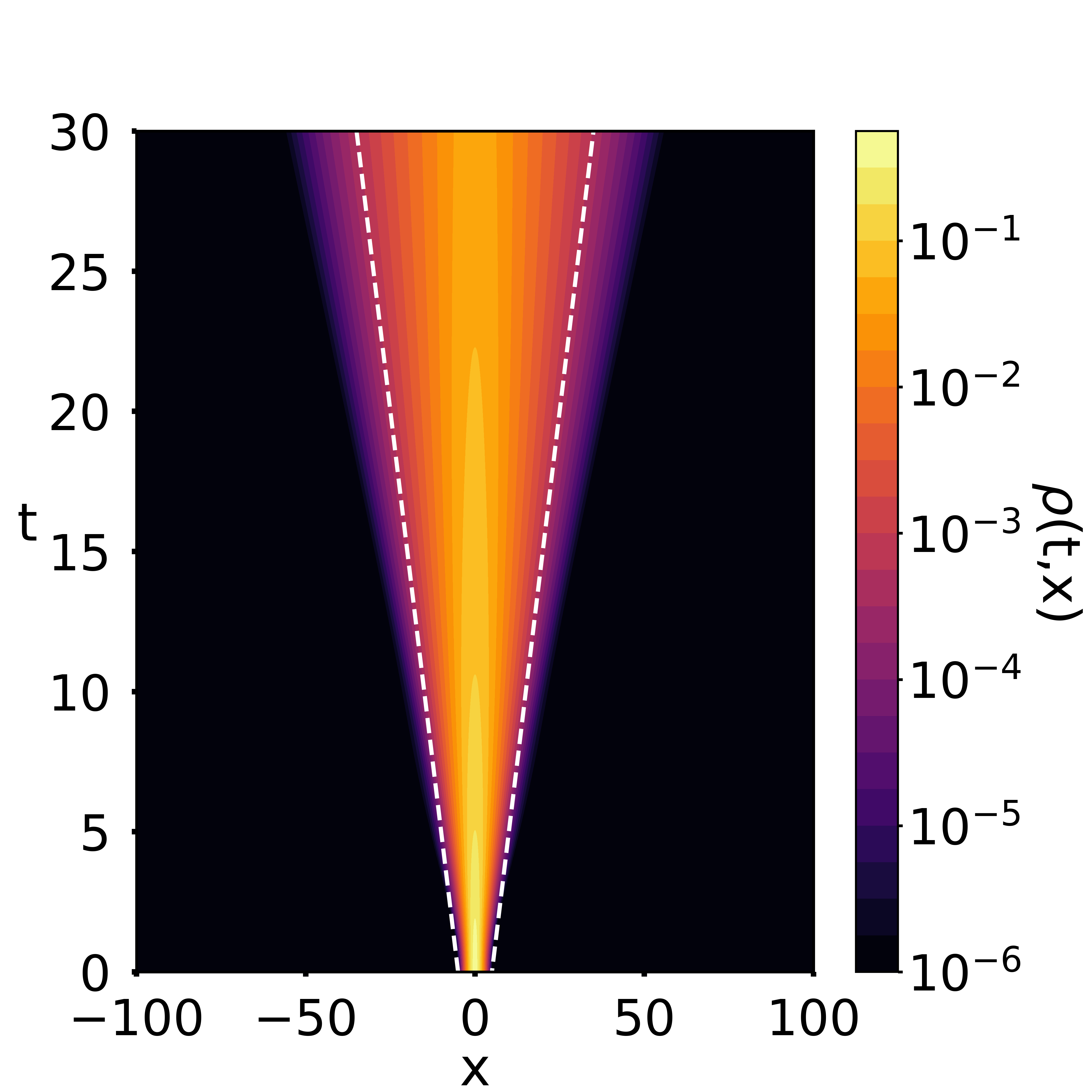}
		        \caption{$\kappa=0.1$}
		        \label{fig:rho_w_0_gauss}
		    \end{subfigure}
		    \hfill
		    \begin{subfigure}{0.32\textwidth}
		    \includegraphics[width=\textwidth]{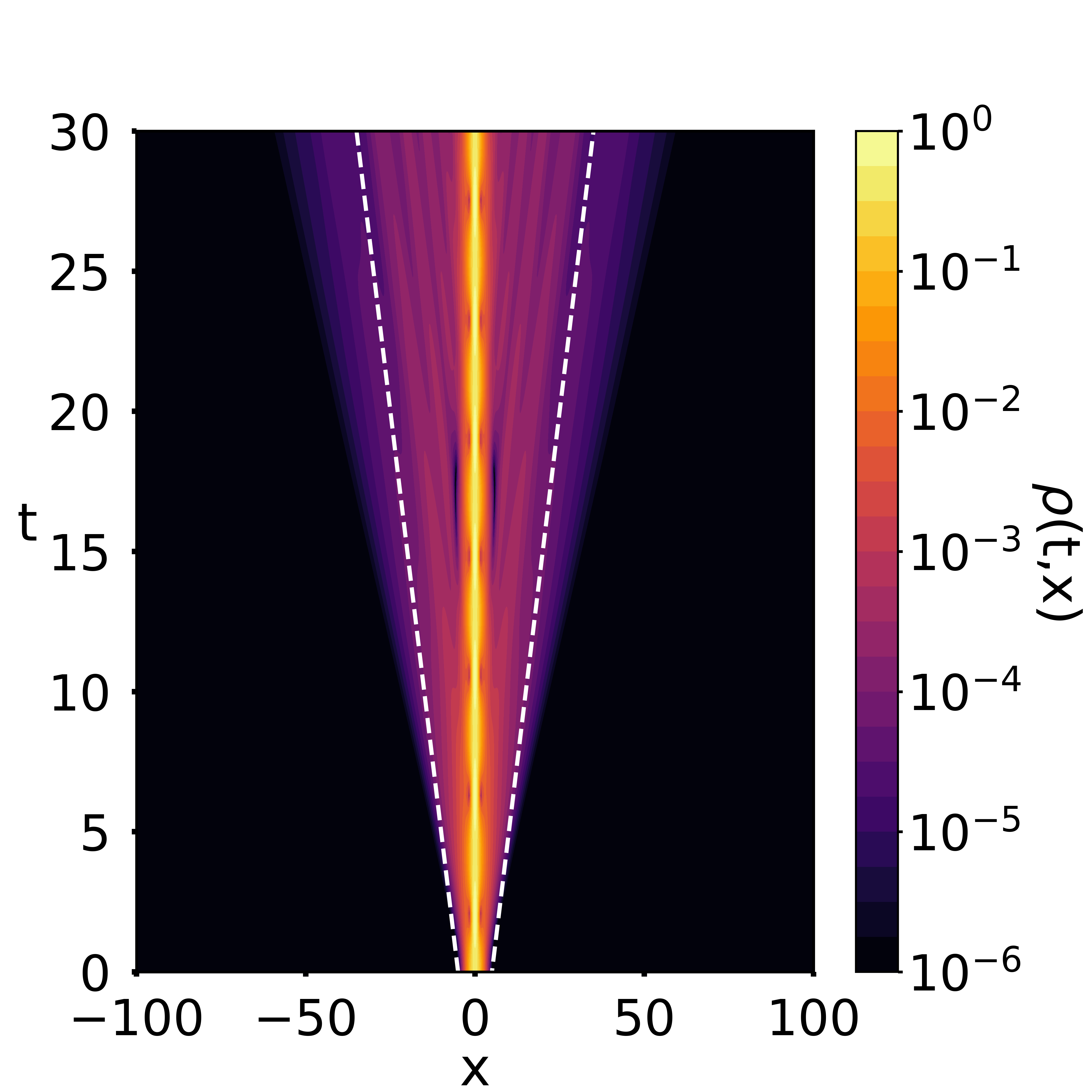}
		        \caption{$\kappa=1$}
		        \label{fig:rho_w_1_gauss}
		    \end{subfigure}
            \hfill
		    \begin{subfigure}{0.32\textwidth}
		      \includegraphics[width=\textwidth]{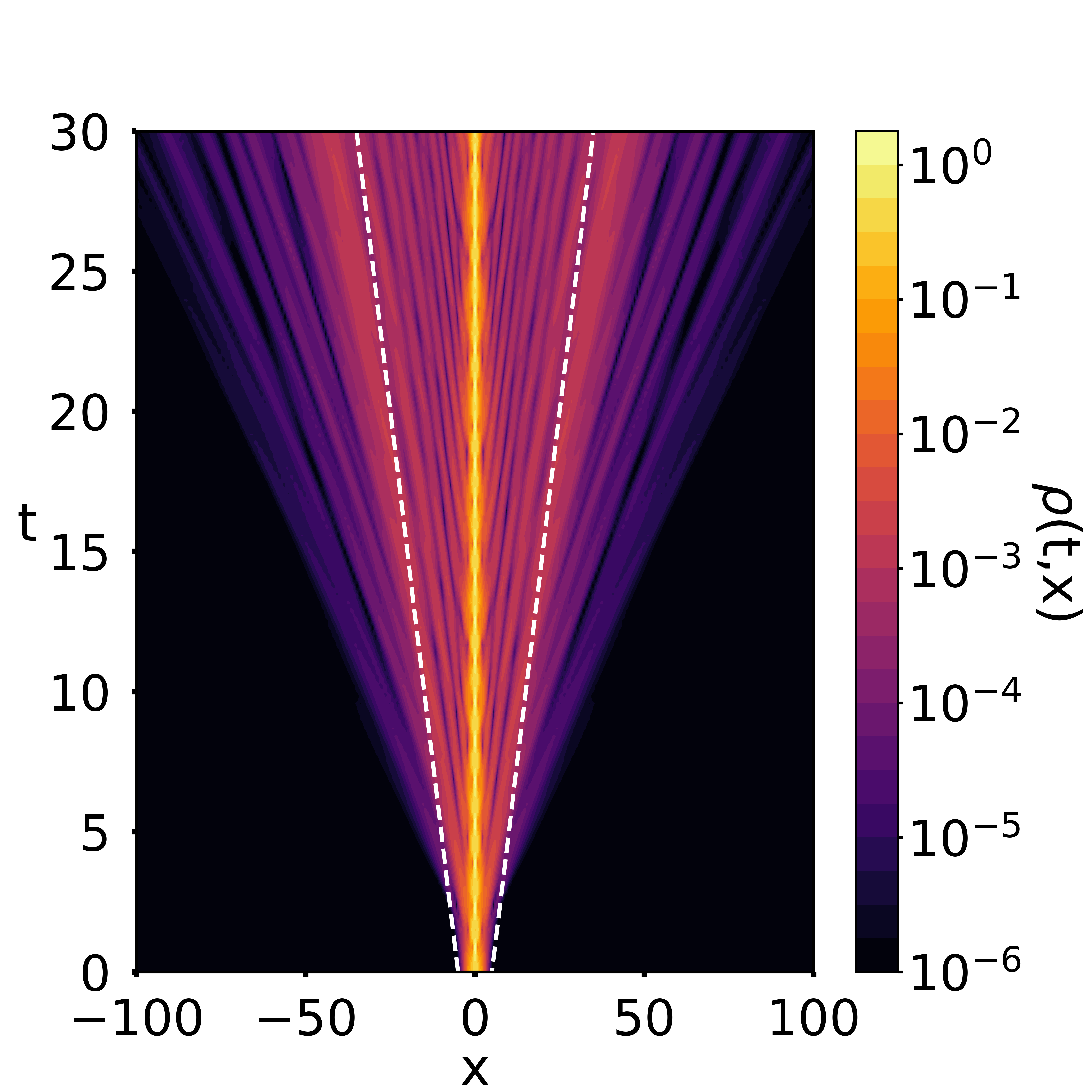}
		        \caption{$\kappa=3$}
		        \label{fig:rho_w_3_gauss}
		    \end{subfigure}
		    \caption{Broader perspective on time evolution of probability density $\rho (t,x)$ for Gaussian initial state. Dashed white lines depict the boundary of the light cone starting starting at $x \in [-5,5]$ and $t=0$.}
		    \label{fig:rho_w_gauss}
		\end{figure}

        During the evolution, the initial Gaussian wavepacket develops a sharply localised core around
its centre together with extended tails at large $|x|$: most of the probability remains in the
central, self-localised part, while a small fraction propagates to infinity.
This effect was observed in the literature on self-gravitating quantum systems in $(3+1)$-dimensions \cite{Harrison2001,Harrison2003,Giulini2011,Meter2011}, but have not been studied quantitatively, to our best knowledge. 
Our simulations show that for larger values of the self-coupling $\kappa$ the amount of probability spreading superluminally \emph{increases} with $\kappa$. 
In order to get a quantitative estimate of the superluminal probability leak we employ the function $M(t,R,\kappa)$ defined in Eq. \eqref{M_int}. The light cones depicted in Fig. \ref{fig:rho_w_gauss} correspond to the choice $R = 5$.  Initially, the probability to find the particle beyond the interval $[-R,R]$ is very small $\approx 5 \cdot 10^{-7}$. It increases monotonically in time, as the tails start to build up beyond the light cone. 
        
        In Fig. \ref{fig:M_gauss} we present the plots of the quantity $M(t,R,\kappa)$ for $\kappa=0$, $\kappa=0.5$, $\kappa=1$ and $\kappa=3$. For $\kappa = 0$ we recover the free Schr\"odinger equation without gravitational self-interaction. Fig. \ref{fig:M_gauss_0} exhibits the same pattern as Fig. 2 in \cite{PRA2017}\footnote{In \cite{PRA2017} the width of the initial Gaussian corresponds to the choice $\sigma = 1/\sqrt{2}$ in Eq. \eqref{Gauss}.} showing that the free nonrelativistic Schr\"odinger equation facilates superluminal signalling \cite{PRA2020}. The maximal amount of superluminal probability leakage is observed for $R \gtrsim 3 \sigma$. For smaller values of $R$ the initial probability to find the particle outside of the interval $[-R,R]$ is substantial and does not get significantly affected by the superluminal spreading. For $R \gg 3 \sigma$ the section of the light cone at time-slice $t$ is large and the outside probability is small. This suggests that while the spreading is superluminal, it has a finite velocity.
        
        Figs. \ref{fig_M_gauss_05} and \ref{fig_M_gauss_1} show that the gravitational self-interaction mitigates the superluminal spreading of the Gaussian wavepacket. This is coherent with the fact that the wavepackets evolving under the Schr\"odinger--Newton equation with a coupling $\kappa$ above $\kappa_c$ tend to localise and hence, naturally, more of the probability density stays within the light cone. 
        However, Fig. \ref{fig_M_gauss_3} shows that this picture breaks down for large $\kappa$. Indeed, the amount of superluminal probability leakage for $\kappa = 3$ is actually (slightly) larger than in the non-interacting case. A comparison between Figs. \ref{fig:rho_w_0_gauss}, \ref{fig:rho_w_1_gauss} and \ref{fig:rho_w_3_gauss} shows that the velocity of superluminal spreading increases with the increasing coupling constant $\kappa$.
        
       The attractive self-interaction compresses the initial Gaussian wavepacket into a small region around its centre, with the localisation length shrinking as $\kappa$ increases. This enhanced localisation broadens the momentum distribution and excites high-momentum modes, which in turn eject part of the wavepacket from the localised region. The resulting escape velocity is finite, yet superluminal, and increases with $\kappa$.

                \begin{figure}[H]
            \centering
            \begin{subfigure}{0.4\textwidth}
            \includegraphics[width=\textwidth]{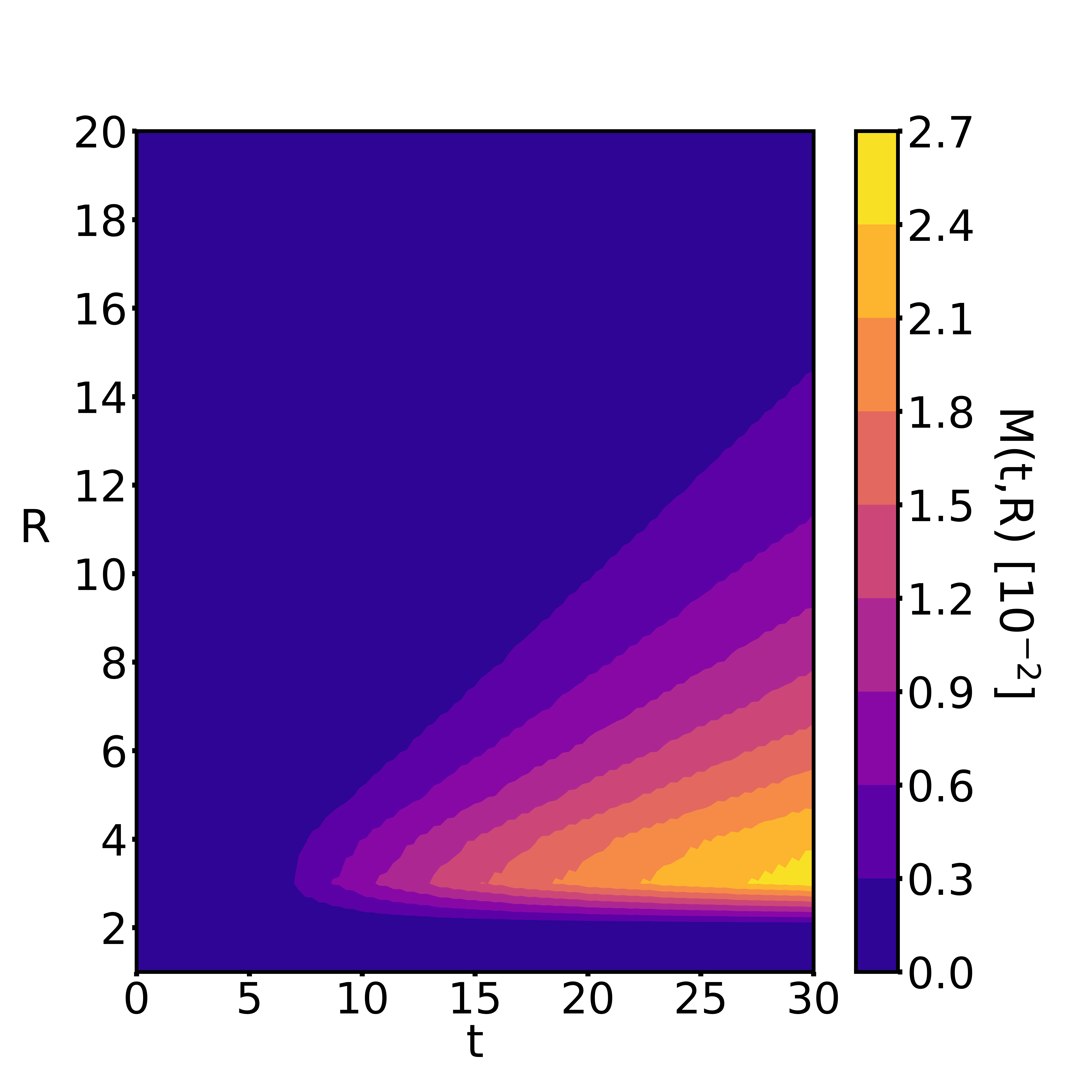}
            \caption{$\kappa = 0$}
            \label{fig:M_gauss_0}
             \end{subfigure}
             \hfill
            \begin{subfigure}{0.4\textwidth}
                \includegraphics[width=\textwidth]{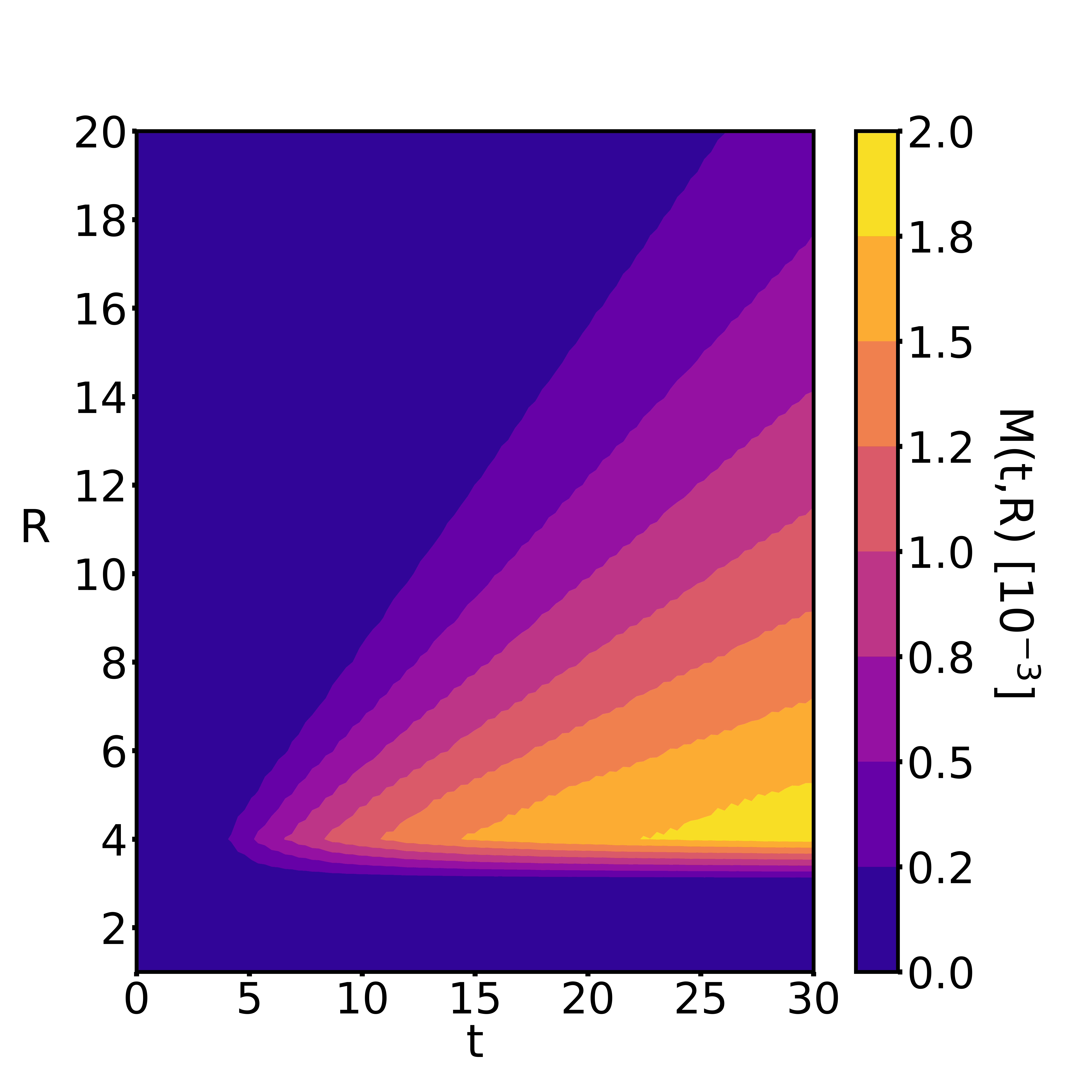}
                \caption{$\kappa=0.5$}
                \label{fig_M_gauss_05}
            \end{subfigure}
            \hfill
            \begin{subfigure}{0.4\textwidth}
                \includegraphics[width=\textwidth]{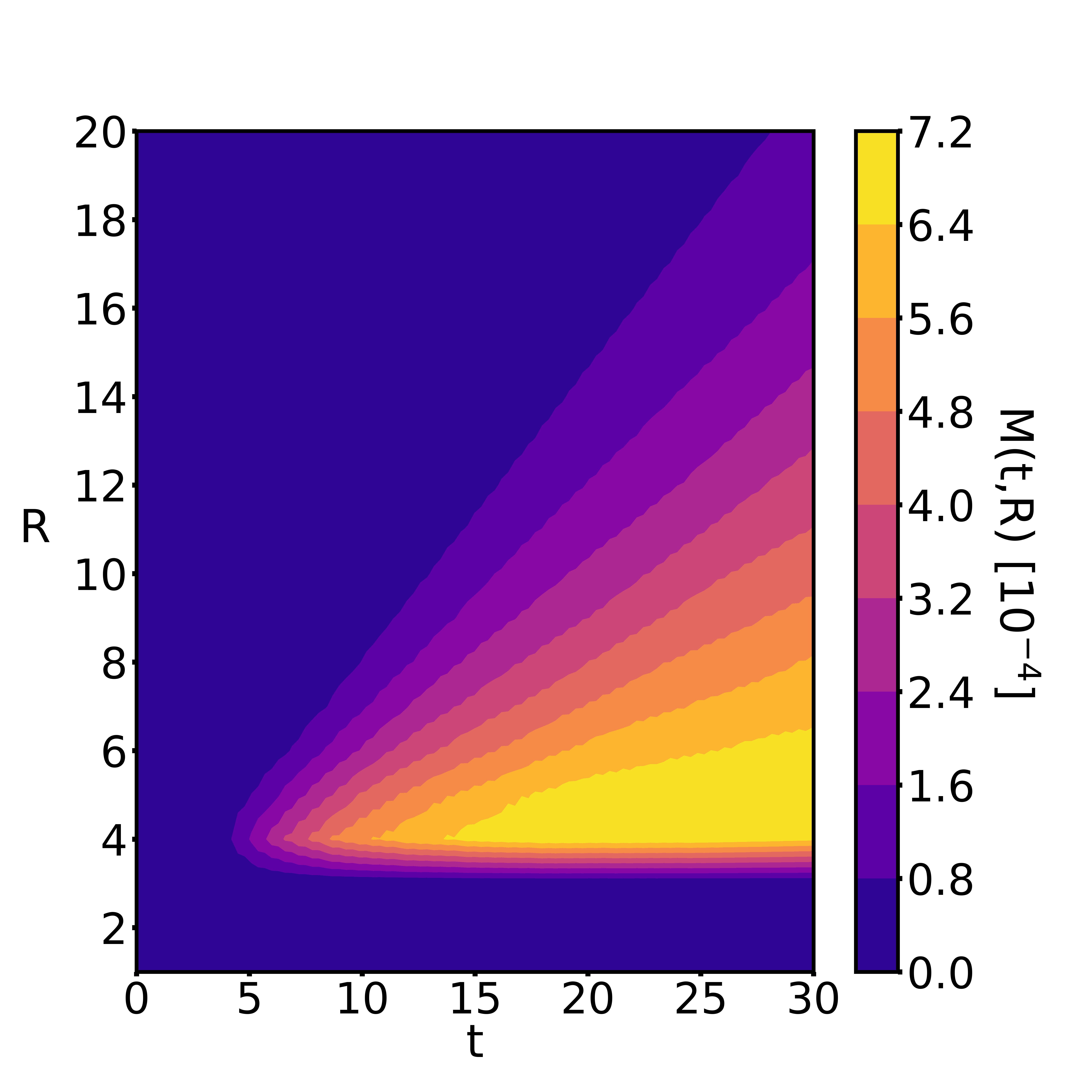}
                \caption{$\kappa=1.0$}
                \label{fig_M_gauss_1}
            \end{subfigure}
            \hfill
            \begin{subfigure}{0.4\textwidth}
                \includegraphics[width=\textwidth]{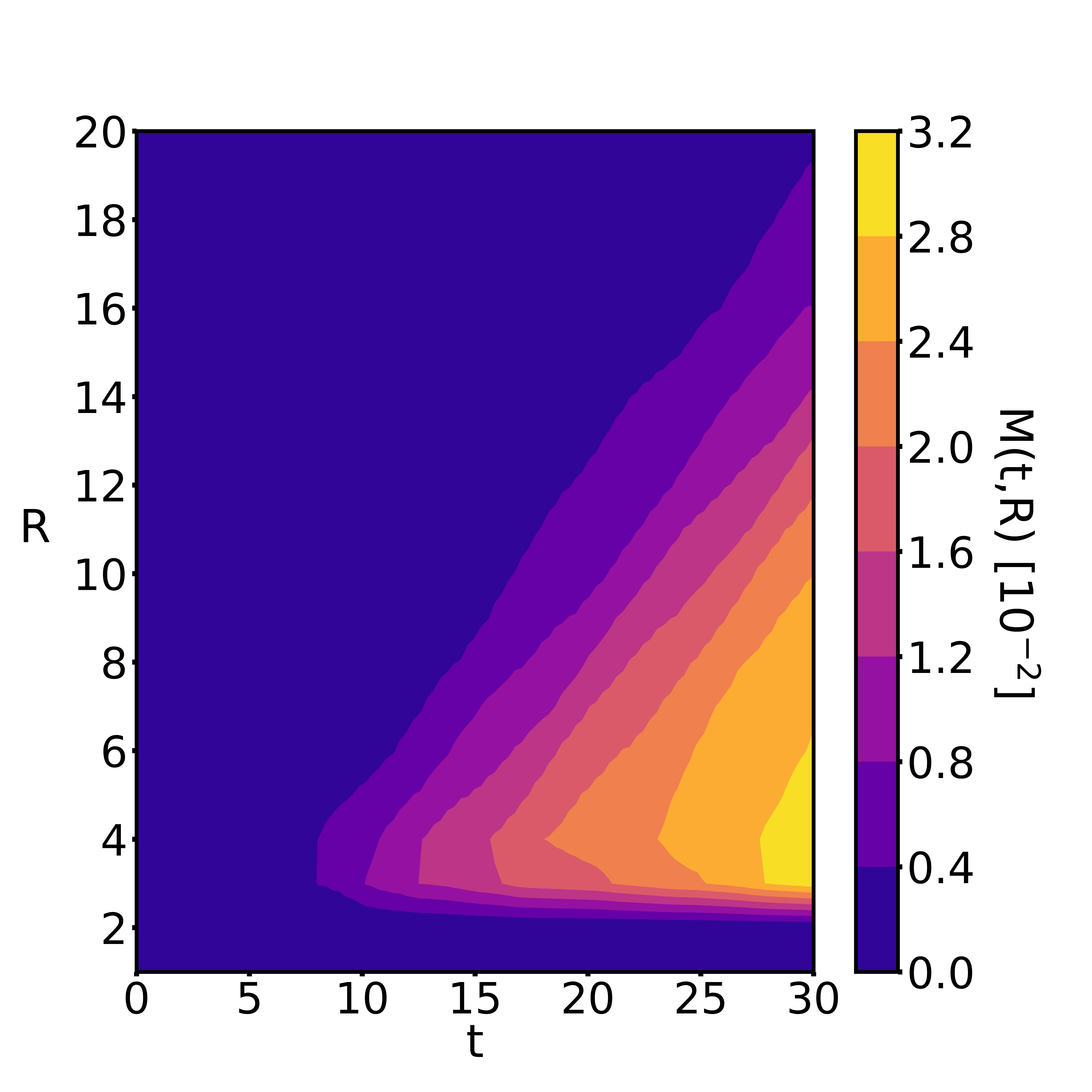}
                \caption{$\kappa=3.0$}
                \label{fig_M_gauss_3}
            \end{subfigure}
            \caption{The plots present the amount of probability density leaking out of the future light cone of the interval $[-R,R]$ at a given time $t$, for a Gaussian initial state evolving under the Schr\"odinger--Newton equation the coupling constant $\kappa$.}
            \label{fig:M_gauss}
        \end{figure}

        \subsection{The trapped particle}\label{sec:trap}

        We now discuss the protocol for operational superluminal signalling described in Sec. \ref{sec:protocol}. To this end, we consider the one-dimensional Schr\"odinger--Newton equation \eqref{SN} with an additional potential: 
		\begin{equation}
		    V(t,x) = 
            \begin{cases}
		        -V_0\, \Theta(-t) & \text{for} \ \ x\in [-R,R], \\
                0 & \text{for} \ \  x\in (-\infty,-R)\cup(R,\infty) ,
		    \end{cases}
            \label{V}
		\end{equation}
		where $\Theta$ is the Heaviside theta function. It is a box potential of depth $V_0$ and width $2R$, which is switched off at $t=0$. Hence, the potential term only determines the shape of the initial wave packet at $t=0$ and we can omit it in the dynamical equations.
        
       For a chosen value of $R$, and self-interaction coupling $\kappa$ we prepare the system in its ground state. The probability density of this state is presented in Fig. \ref{fig:rho_gs_INI}. For a finite depth of the potential well the state will have some tails outside of the trapping region. This is not an issue for our study, because the operational measure of superluminal signalling \eqref{M} quantifies the difference between the detection probabilities with and without the trap (recall Sec. \ref{sec:protocol}). For the numerical simulations we have set $V_0=20$.

      Two parameters of the system, i.e. $R$ and $\kappa$ determines the properties of the ground state. In particular, strong enough self-interaction $\kappa$ results in a self-bound ground state, remaining stable after releasing the trap. The quantitative picture is given by the phase diagram in the $\kappa - R$ plane of  the ground state interaction to kinetic energy ratio, $\delta = \frac{E_{\rm int}}{E_{\rm kin}}$, Fig. \ref{fig:E_ratio_SN}, where 
         \begin{align}\label{energy}
        E_{\text{kin}} = \frac{1}{2} \int dx \big\vert \partial_x \psi_\text{GS}(x) \big\vert^2, &&   E_{\text{int}} = - \frac{\kappa}{2} \int dxdy\,  \frac{\vert \psi_\text{GS}(y) \vert^2 \vert \psi_\text{GS}(x) \vert^2}{\vert x- y \vert}.   
         \end{align}
      For $|\delta|>1$, the interaction energy dominates kinetic energy, and state is in the self-bound regime.

        \begin{figure}[H]
		    \centering
		    \begin{subfigure}[t]{0.5\textwidth}
		        \includegraphics[width=\textwidth]{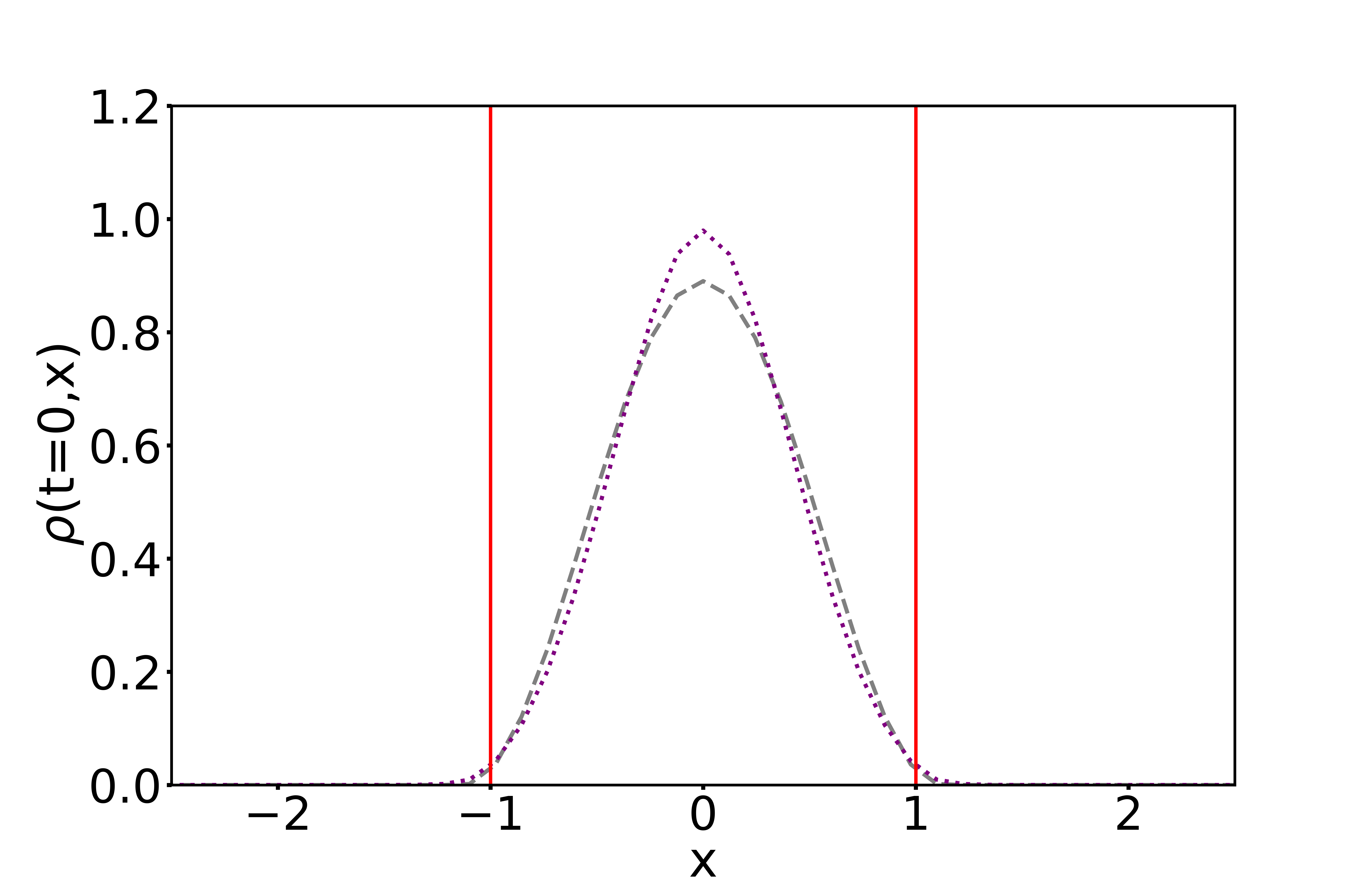}
		        \caption{
                Dotted line shows the probability density of initial ground states of Eq. \eqref{eq:SN_dim} with the box potential \eqref{V}, with $R=1$, $V_0=-20$ for $\kappa = 0$ (dashed line) and for $\kappa = 1$ (dotted line).
                The solid lines mark the boundary of the  trap.}
		        \label{fig:rho_gs_INI}
		    \end{subfigure}
            \hfill
            \begin{subfigure}[t]{0.43\textwidth}
		    \includegraphics[width=\textwidth]{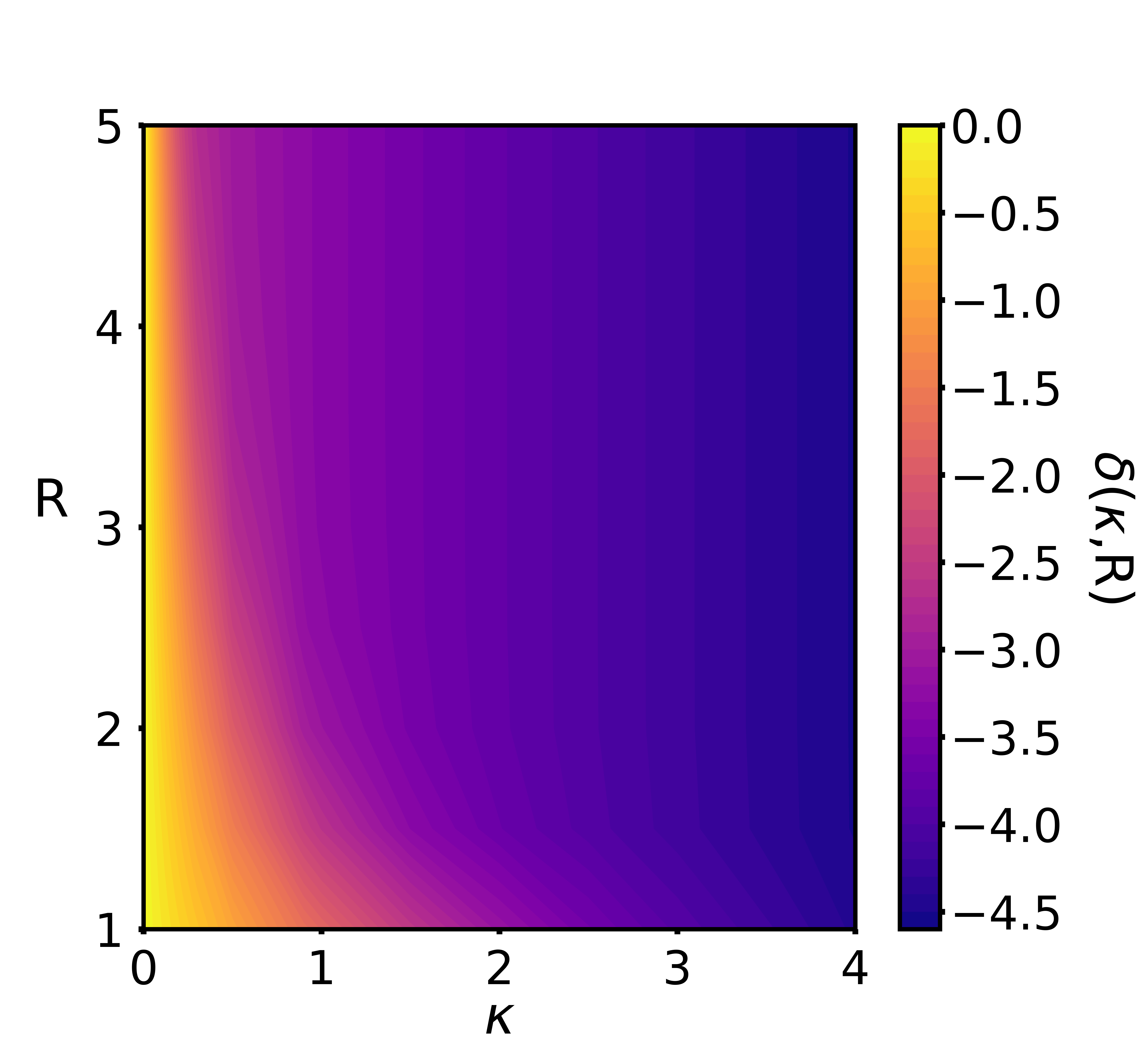}
		        \caption{The ratio of the interaction energy to kinetic energy \eqref{energy} of the initial ground state in the trap as a function of the coupling constant $\kappa$ and the trap size $R$.}
		        \label{fig:E_ratio_SN}
		    \end{subfigure}
            \caption{The structure of the ground state of Eq. \eqref{eq:SN_dim} with the box potential \eqref{V}.}
		\end{figure}

        The ground state is then released from the trap at $t=0$ and evolves according to Eq. \eqref{eq:SN_dim}. In Fig. \ref{fig:rho_gs_snap} we compare the time evolution of the initial state with $(\kappa = 1)$ and without $(\kappa = 0)$ the gravitational self-interaction. It is evident that the self-coupling leads to wavepacket localisation within the region of the initial trap. Fig. \ref{fig:rho_gs_comp} also shows that the amount of probability, which leaks out of the causal cone for $\kappa = 1$ is roughly more than an order of magnitude smaller than in the free case.

         \begin{figure}[H]
		    \centering
		    \begin{subfigure}[t]{0.48\textwidth}
		        \includegraphics[width=\textwidth]{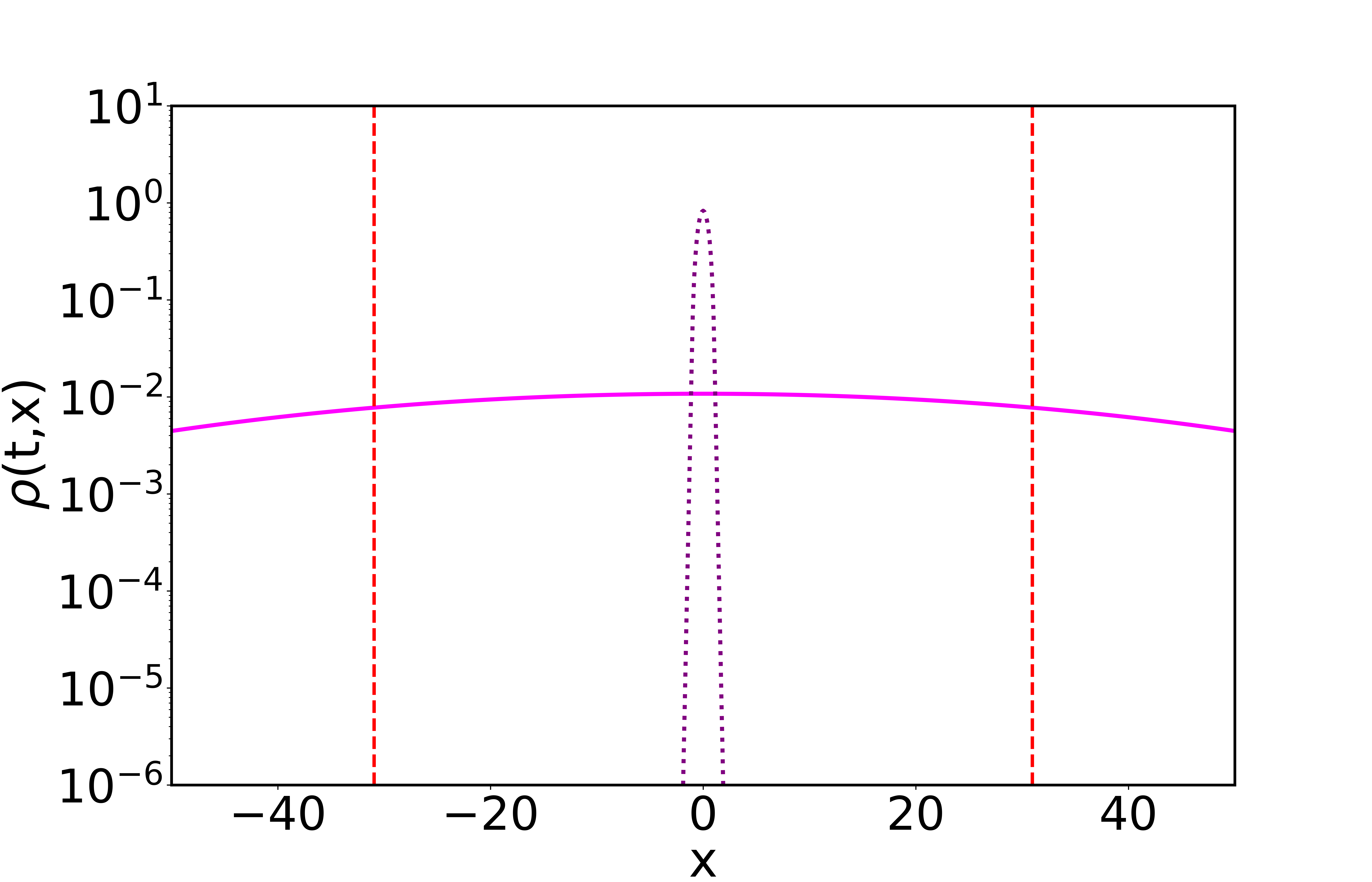}
		        \caption{Solid line presents the time-evolved probability density at $t=30$ of the initial ground state for $\kappa=0$. The dotted line shows the initial probability density.}
		        \label{fig:rho_gs_k0}
		    \end{subfigure}
		    \hfill
		    \begin{subfigure}[t]{0.48\textwidth}
		        \includegraphics[width=\textwidth]{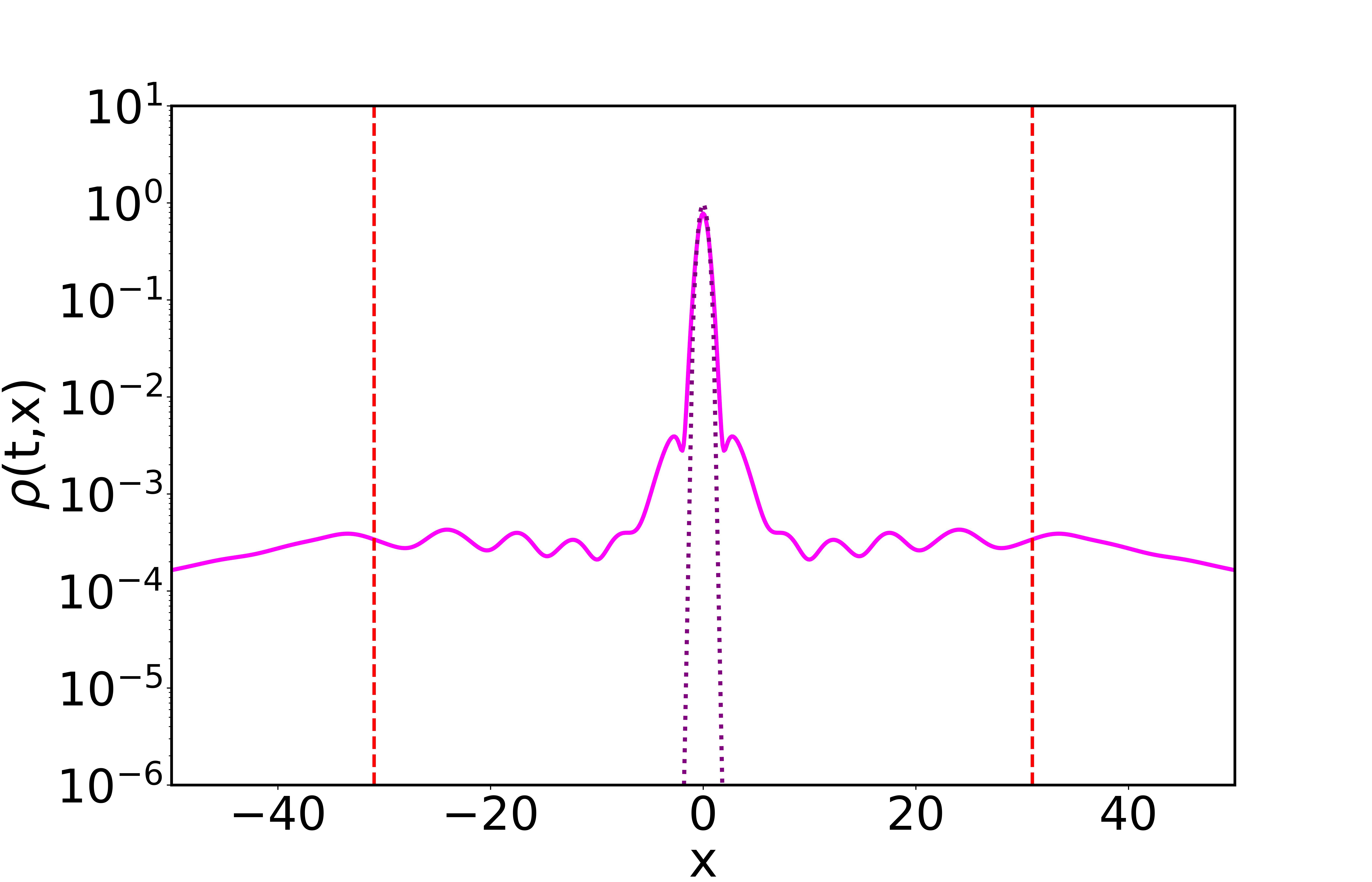}
		        \caption{Solid line presents the time-evolved probability density at $t=30$ of the initial ground state for $\kappa=1$. The dotted line shows the initial probability density.}
		        \label{fig:rho_gs_comp}
		    \end{subfigure}

		    \caption{Comparison of time evolved probability densities resulting from Eq. \eqref{eq:SN_dim} with the box potential \eqref{V}. Dashed lines mark the boundary of the light cone of the initial trapping region.}
		    \label{fig:rho_gs_snap}
		\end{figure}

         Figs. \ref{fig:rho_gs} and \ref{fig:rho_w_gs} present the time-evolution of the probability density induced by Eq. \eqref{eq:SN_dim} with the initial state taken to be the ground state in the box potential. Fig.~\ref{fig:rho_gs}  demonstrates that the wave packet tends to localise as the coupling constant $\kappa$ increases. In contrast to the Gaussian case, there is no superluminal outflow of probability density for large $\kappa$. This is because, for large $\kappa$ the ground states is self-stabilised, even when the trap is switched off. This is particularly visible in Figs. \ref{fig:rho_w_1_gs} and~\ref{fig:rho_w_3_gs}.

        \begin{figure}[H]
		    \centering
		    \begin{subfigure}{0.32\textwidth}
		    \includegraphics[width=\textwidth]{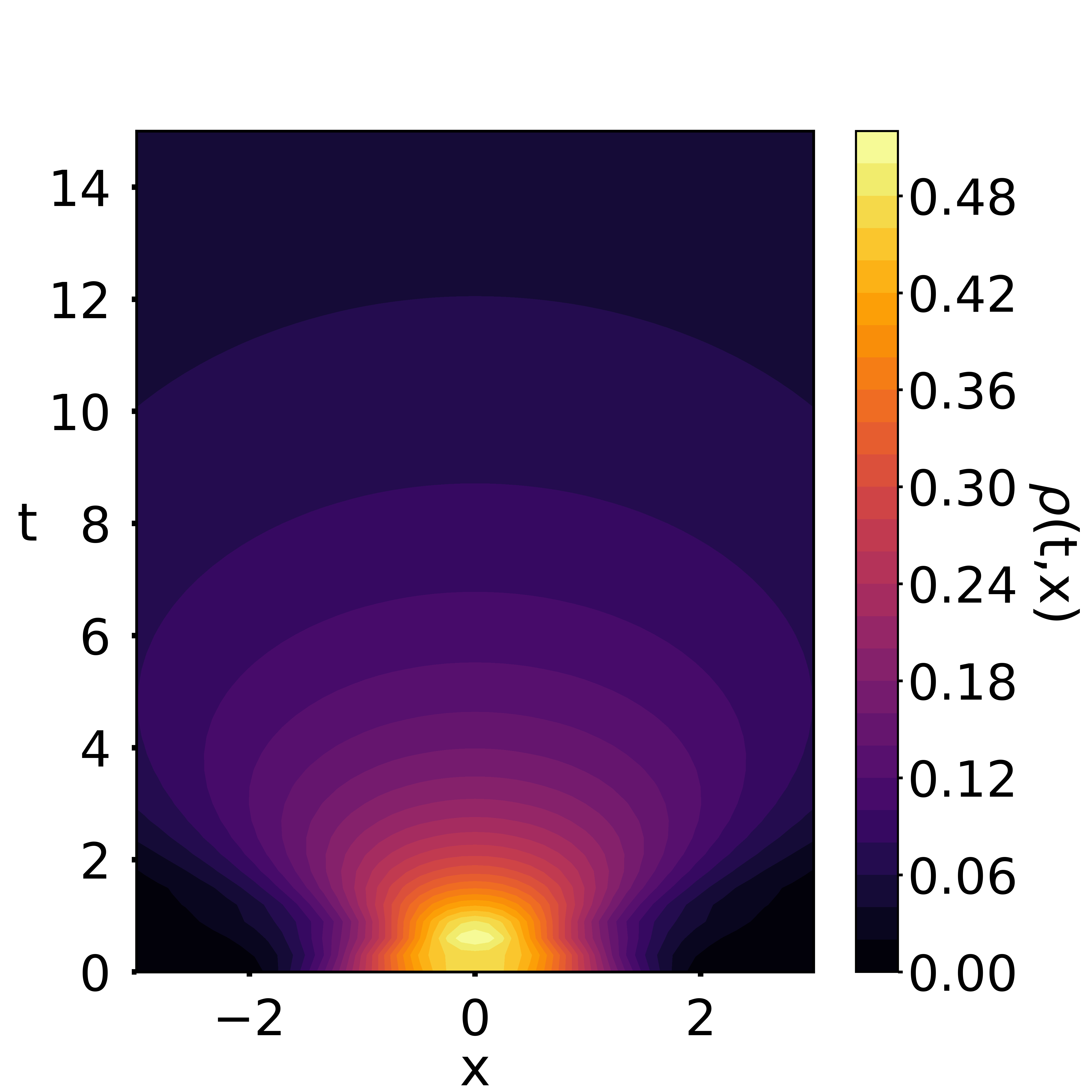}
		        \caption{$\kappa=0.1$}
		        \label{fig:rho_01_gs}
		    \end{subfigure}
		    \hfill
		    \begin{subfigure}{0.32\textwidth}
		      \includegraphics[width=\textwidth]{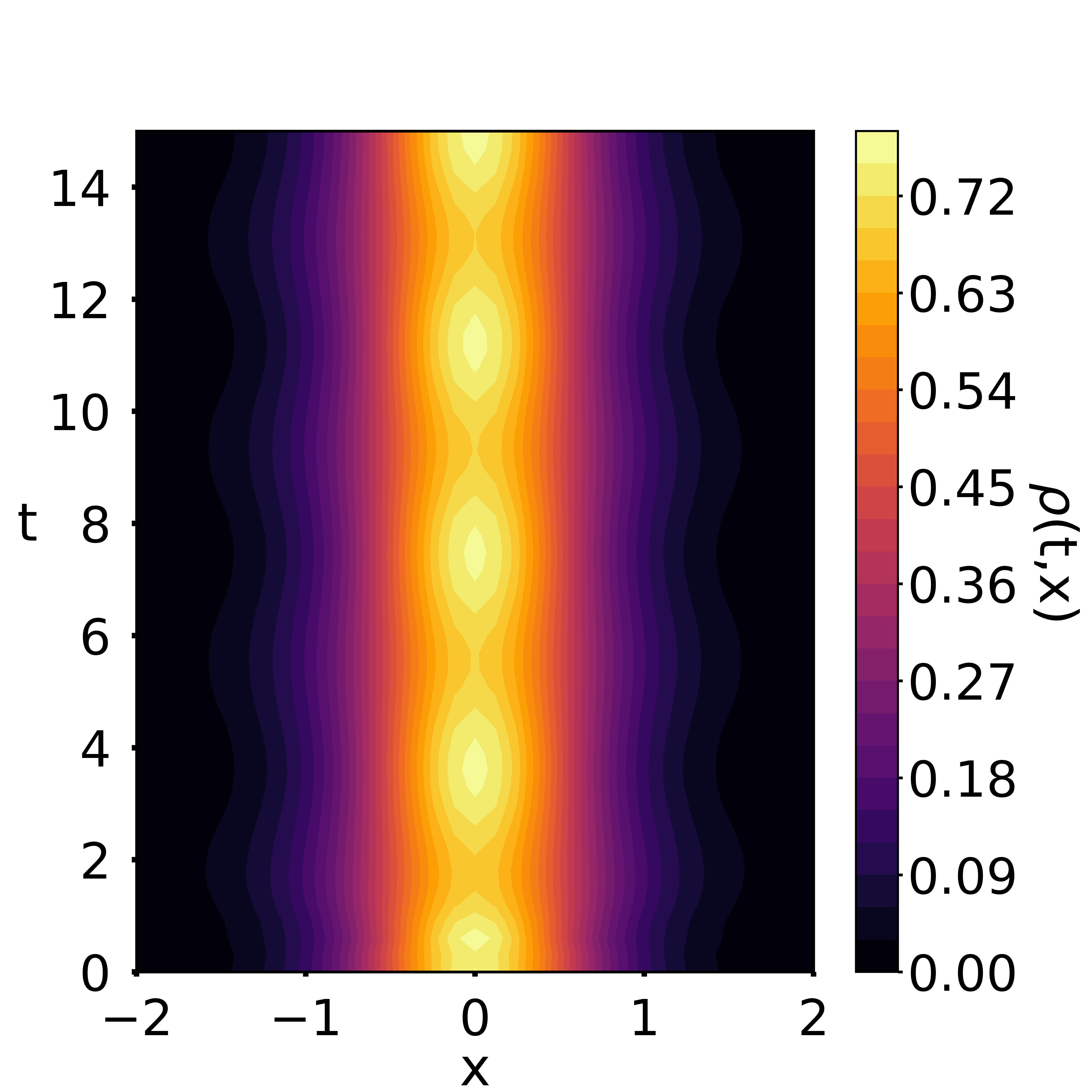}
		        \caption{$\kappa=1$}
		        \label{fig:rho_1_gs}
		    \end{subfigure}
            \hfill
		    \begin{subfigure}{0.32\textwidth}
		    \includegraphics[width=\textwidth]{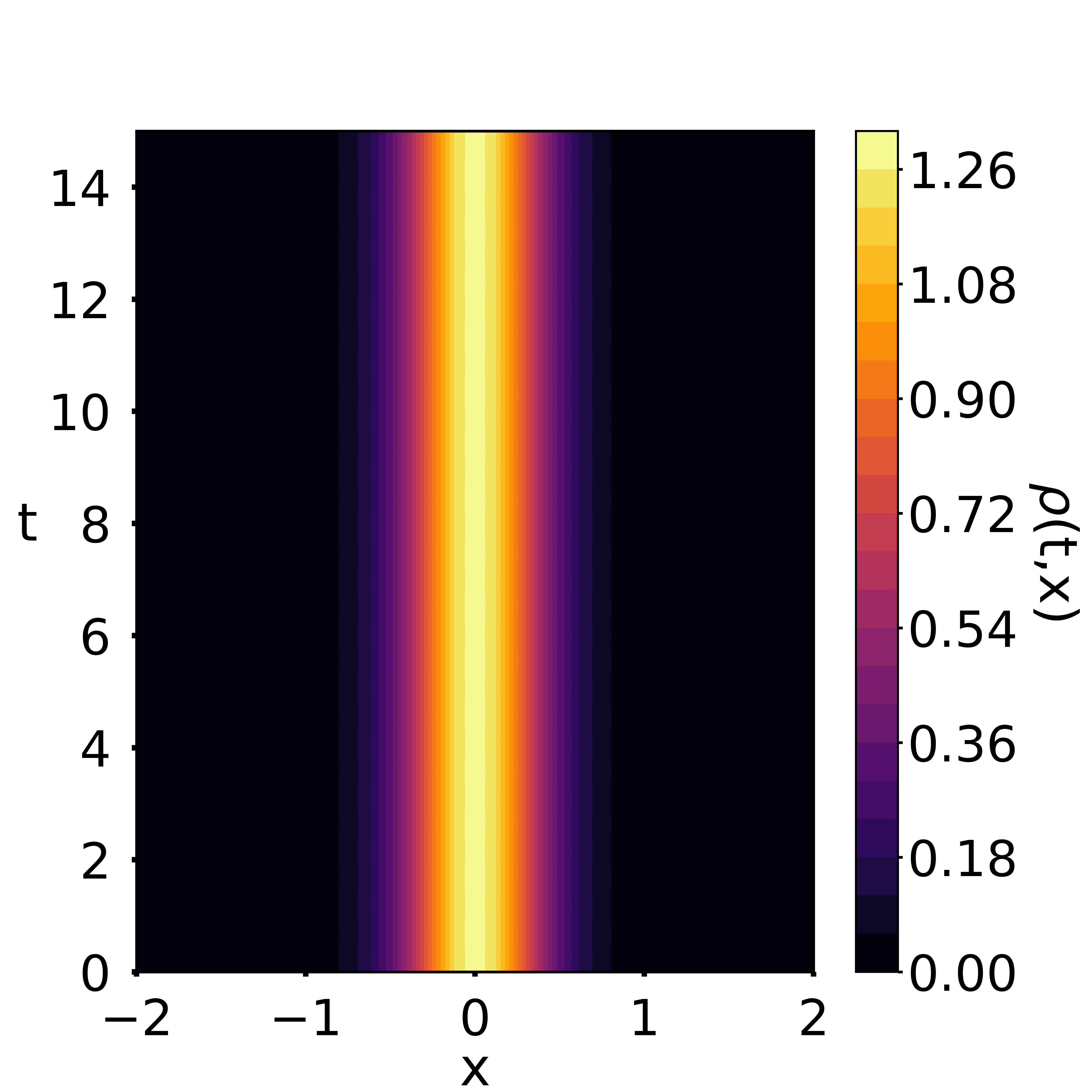}
		        \caption{$\kappa=3$}
		        \label{rho_3_gs}
		    \end{subfigure}
		    \caption{Time evolution of probability density $\rho (x,t)$ for initial ground state for $R=2$ and $V_0=-20$.}
		    \label{fig:rho_gs}
		\end{figure}

        \begin{figure}[H]
		    \centering
		    \begin{subfigure}{0.32\textwidth}
		    \includegraphics[width=\textwidth]{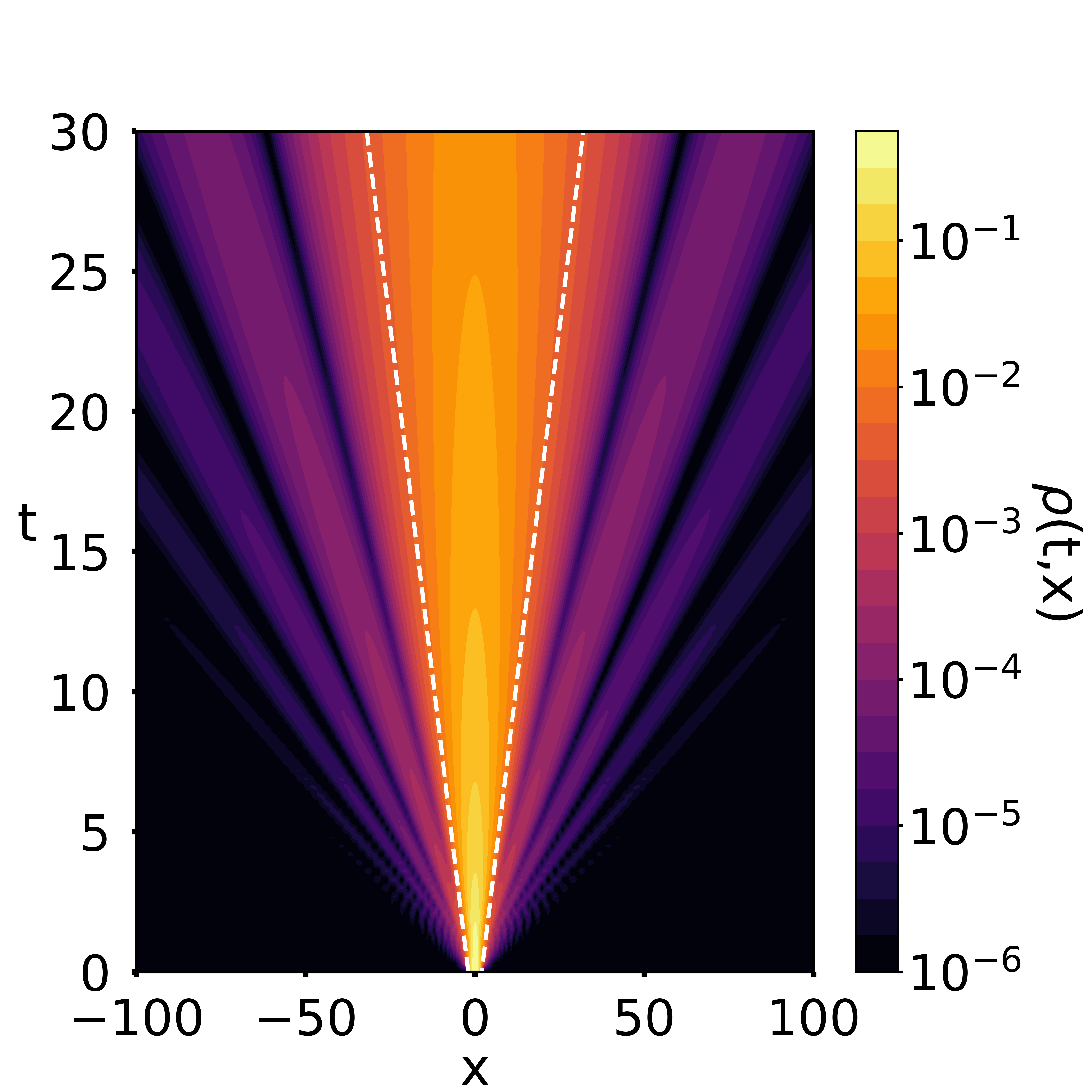}
		        \caption{$\kappa=0.1$}
		        \label{fig:rho_w_01_gs}
		    \end{subfigure}
		    \hfill
		    \begin{subfigure}{0.32\textwidth}
		      \includegraphics[width=\textwidth]{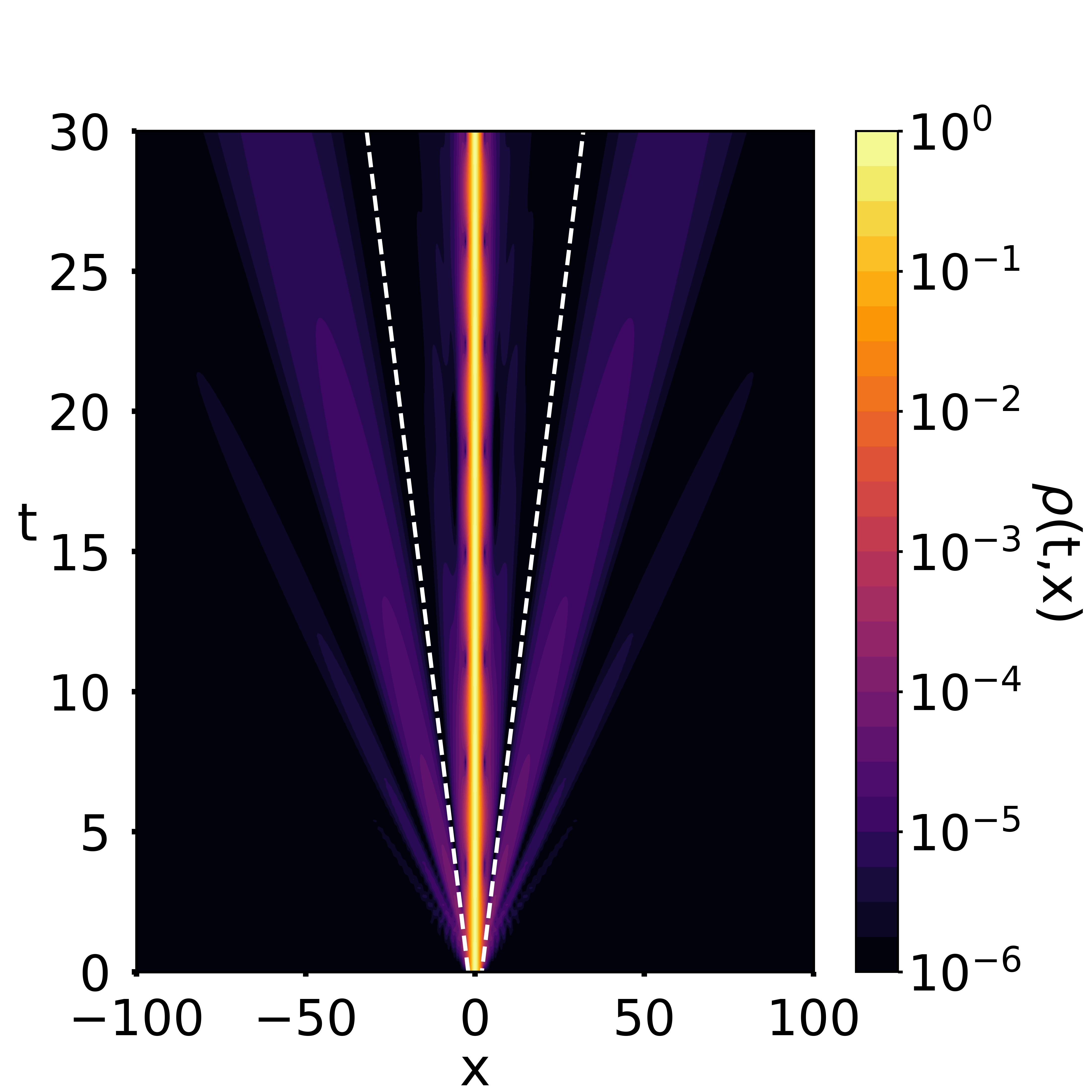}
		        \caption{$\kappa=1$}
		        \label{fig:rho_w_1_gs}
		    \end{subfigure}
            \hfill
		    \begin{subfigure}{0.32\textwidth}
		    \includegraphics[width=\textwidth]{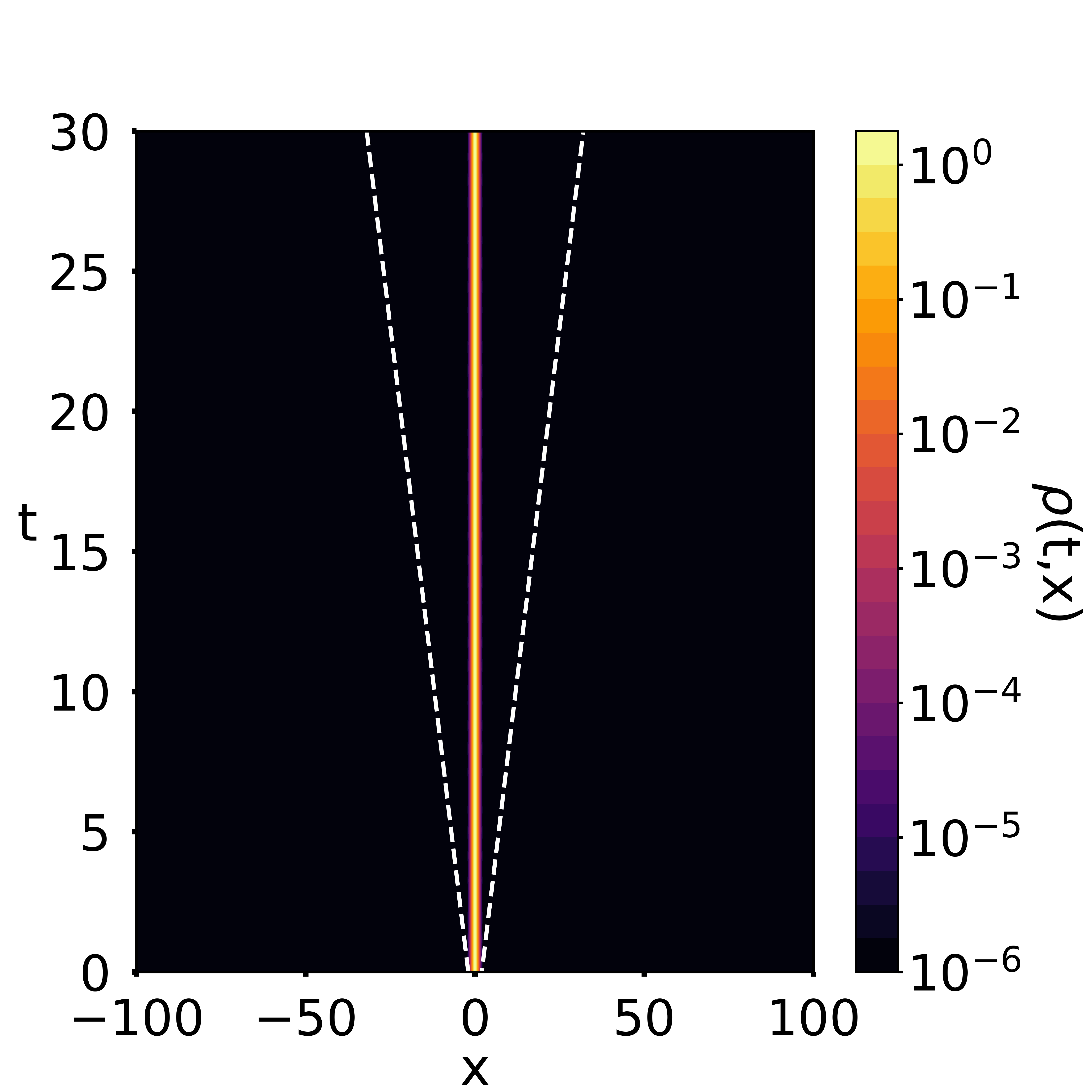}
		        \caption{$\kappa=3$}
		        \label{fig:rho_w_3_gs}
		    \end{subfigure}
		    \caption{Broader perspective on time evolution of probability density $\rho (x,t)$ for initial ground state for $R=2$ and $V_0=-20$. White dashed lines denote the future light cone of the initial trapping region $x \in [-R,R]$.}
		    \label{fig:rho_w_gs}
		\end{figure}

         In order to quantify the amount of probability, which spreads superluminaly we employ the function \eqref{M_max}. In the present context, the quantity $\widetilde{M}(\kappa,R)$ has a natural operational interpretation --- see Sec. \ref{sec:protocol}. Namely, it is the maximal success probability of superluminal bit transfer over time $t \in [0,30]$, for the trap of size $R$ and a given self-coupling strength $\kappa$. 

         In Fig. \ref{fig:M_max_SN} we present how the quantity $\widetilde{M}(\kappa,R)$ changes for a range of $R\in[1,5]$ with step $\Delta R=0.5$ and $\kappa\in[0,4]$ with a non-uniform step. For $\kappa\in[0,1]$ the step is $\Delta \kappa = 0.01$, for $\kappa\in[1,3]$ $\Delta \kappa = 0.02$ and for $\kappa>3$ we have $\Delta \kappa = 1$. It shows that the bigger the trap and the more massive the system, the smaller the quantity $\widetilde{M}$. This is coherent with the phase diagram for the initial state (Fig. \ref{fig:E_ratio_SN}).
         
         We thus clearly see that self-gravity actually mitigates the superluminal effects of the wavepacket spreading. In this sense, the Schr\"odinger--Newton equation is in fact `more compatible' with relativistic causality than the ordinary free Schr\"odinger equation.

		\begin{figure}[H]
            \centering
                \resizebox{0.5\textwidth}{!}{\includegraphics{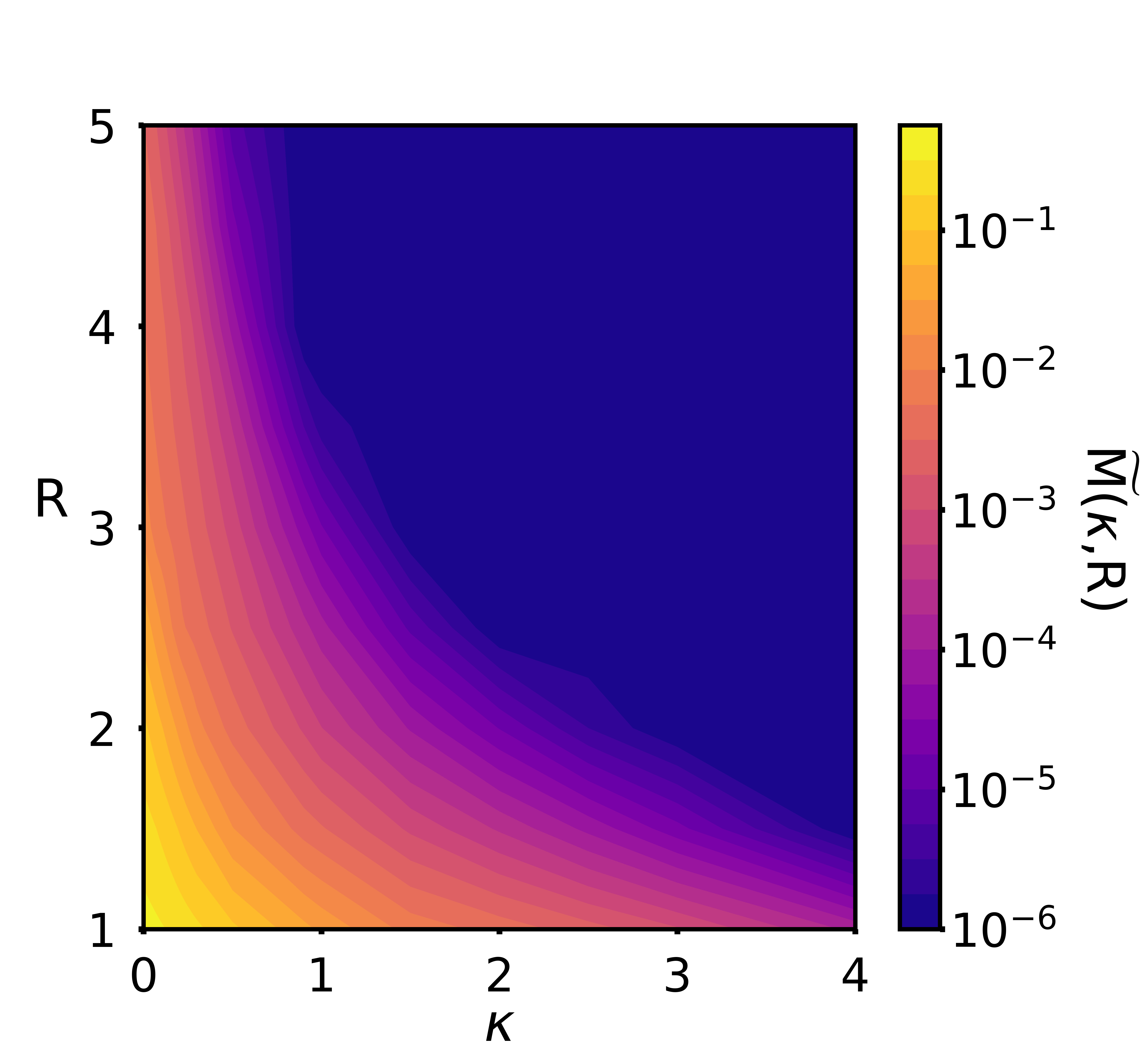}}
                \caption{
                Quantification of superluminal signalling in the Schr\"odinger--Newton model in function of the self-coupling strength and the trap size.
                }
            \label{fig:M_max_SN}
        \end{figure}

\section{Summary and outlook}
\label{sec:sum}

The Schr\"odinger--Newton (SN) equation attempts to describe the behaviour of mesoscopic quantum systems, in which the gravitational force may become relevant. It predicts self-localisation for wave packets of micrometre size and masses of order $10^{10}$\,AMU, a regime that remains experimentally challenging but may become testable with future macroscopic-quantum technologies \cite{OR,QG_RMP,Penrose03,LimitsQM,Howl2019,SN2022}.

A frequent criticism is that SN dynamics enables superluminal signalling \cite{OR,SN_review_14}. Our main message is that this should be assessed operationally and scale-wise, rather than treated as an a priori inconsistency: operational superluminal signalling can already arise for the free Schr\"odinger equation (Sec.~\ref{sec:protocol}) \cite{PRA2017,PRA2020}, which is nonetheless an accepted nonrelativistic model. In this work we therefore focused on quantifying signalling scales and on relating the SN approximation to a fully relativistic setting. In particular, we proved that the Einstein--Dirac system is rigorously consistent with relativistic causality, and we found (qualitatively) that including gravitational self-coupling can be ``less signalling'' than the decoupled free Schr\"odinger evolution.

A central open problem is the extension from the present single-particle setting to coupled multiparticle and/or field-theoretic dynamics, where Gisin-type concerns about nonlinear quantum dynamics become relevant \cite{Gisin1989,Gisin2001}. We advocate a pragmatic viewpoint, consistent with our earlier work \cite{PRA2017,PRA2020}: one should determine the characteristic magnitudes, distances, and times associated with operational signalling in concrete models. The qualitative two-particle SN analysis in \cite{SN_review_14} aligns with this perspective, suggesting that signalling is possible in principle but negligible in practice. Finally, superluminal signalling is not the only structural difficulty for nonlinear composite-system models: nonlinear interactions clash with the standard tensor-product structure \cite{Mielnik01} and with the usual mixed-state description \cite{Diosi25}. These issues motivate, rather than preclude, careful effective modelling.

Looking ahead, three directions seem particularly natural:

\begin{itemize}
\item[(i)] \textbf{Cosmological disorder and Anderson-type suppression of spreading.}
Primordial density perturbations are well described by an approximately Gaussian random field whose power spectrum seeds large-scale structure \cite{Pee20,BBKS86,Dodelson03,Mukhanov05}. In a Newtonian sub-horizon regime, the associated gravitational potential can be viewed as a correlated random potential for SN dynamics \cite{Pee20,Dodelson03,Mukhanov05}. Since disorder induces Anderson localization in the linear single-particle limit \cite{And58,AAL79} and many-body localization in interacting systems \cite{SLS25}, it is natural to ask whether cosmologically motivated randomness can hinder long-time spreading in SN dynamics and thereby suppress large-distance/late-time signalling. Related wave-physics intuition comes from stochastic lensing/scattering and wave-optics effects for waves propagating through inhomogeneous cosmological matter distributions \cite{TakahashiNakamura03,Macquart07}. Because SN is nonlinear and self-consistent, the outcome may resemble the broader phenomenology of nonlinear disordered wave equations (localization versus slow subdiffusion) \cite{PS08,FKS09}, and deserves a systematic study. (In a different, higher-dimensional context, Anderson-type localization has even been invoked to localize gravity itself in disordered setups \cite{Rothstein13}.)

\item[(ii)] \textbf{Nonlinear effective many-body dynamics and Lieb--Robinson violations.}
While microscopic many-body quantum dynamics is linear and constrained by Lieb--Robinson bounds \cite{LB72}, practical approximations (mean field, Hartree/Hartree--Fock, Gutzwiller, tensor networks, etc.) typically produce nonlinear Schr\"odinger-type effective equations \cite{LSA12}. In direct analogy with the present work, it is natural to quantify how such nonlinear effective descriptions can exhibit operational violations of Lieb--Robinson bounds, how these depend on parameters, and how they correlate with the ``classicality'' of the effective dynamics. Connections to Bell-type and Leggett--Garg-type scenarios are already emerging in our recent studies \cite{TL25,MFL25}.

\item[(iii)] \textbf{Analogue platforms in ultracold atoms.}
SN-type mean-field structure also arises in ultracold gases where engineered fields induce an attractive $1/r$ interaction \cite{PhysRevLett.84.5687,PhysRevA.63.031603}. Such systems realise self-bound condensates and SN-analogous collective dynamics \cite{Giovanazzi_2001,PhysRevA.65.053616,PhysRevA.76.053604,PhysRevA.78.013615}, with extensions to $1/r^b$ interactions, multicomponent settings, and reduced-dimensional or finite-temperature regimes \cite{PhysRevA.82.023615,Li2014,PENG20167,10.1063/10.0000130}. Recent work on tightly self-trapped 3D modes and pattern-forming excitations further supports their use as an accessible laboratory testbed for SN-motivated phenomena \cite{PhysRevE.111.064207,liu2025detectingcollectiveexcitationsselfgravitating}.
\end{itemize}

\section{Acknowledgements}
{We thank Tomasz Miller for his help with shaping Theorem \ref{thm:ED}. J.O. and M.E. acknowledge support from the National Science Centre in Poland under the research grant Sonata BIS \textit{Beyond Quantum Gravity} (2023/50/E/ST2/00472).
M.L. acknowledges support from:
European Research Council AdG NOQIA; MCIN/AEI (PGC2018-0910.13039/501100011033, CEX2019-000910-S/10.13039/501100011033, Plan National FIDEUA PID2019-106901GB-I00, Plan National STAMEENA PID2022-139099NB, I00, project funded by MCIN/AEI/10.13039/501100011033 and by the ``European Union NextGenerationEU/PRTR'' (PRTR-C17.I1), FPI); QUANTERA DYNAMITE PCI2022-132919, QuantERA II Programme co-funded by European Union’s Horizon 2020 program under Grant Agreement No 101017733; Ministry for Digital Transformation and of Civil Service of the Spanish Government through the QUANTUM ENIA project call - Quantum Spain project, and by the European Union through the Recovery, Transformation and Resilience Plan - NextGenerationEU within the framework of the Digital Spain 2026 Agenda; MICIU/AEI/10.13039/501100011033 and EU (PCI2025-163167);
Fundació Cellex; Fundació Mir-Puig; Generalitat de Catalunya (European Social Fund FEDER and CERCA program;
}

\bibliographystyle{naturemag}
\bibliography{SN}

@book{Pee20,
  author    = {Peebles, P. J. E.},
  title     = {Principles of Physical Cosmology},
  publisher = {Princeton University Press},
  address   = {Princeton, NJ},
  year      = {2020},
  series    = {Princeton Series in Physics},
  isbn      = {9780691019338}
}

@article{And58,
  author  = {Anderson, P. W.},
  title   = {Absence of Diffusion in Certain Random Lattices},
  journal = {Physical Review},
  volume  = {109},
  number  = {5},
  pages   = {1492--1505},
  year    = {1958},
  doi     = {10.1103/PhysRev.109.1492}
}

@article{AAL79,
  author  = {Abrahams, E. and Anderson, P. W. and Licciardello, D. C. and Ramakrishnan, T. V.},
  title   = {Scaling Theory of Localization: Absence of Quantum Diffusion in Two Dimensions},
  journal = {Physical Review Letters},
  volume  = {42},
  number  = {10},
  pages   = {673--676},
  year    = {1979},
  doi     = {10.1103/PhysRevLett.42.673}
}

@article{SLS25,
  author        = {Sierant, Piotr and Lewenstein, Maciej and Scardicchio, Antonello and Vidmar, Lev and Zakrzewski, Jakub},
  title         = {Many-Body Localization in the Age of Classical Computing},
  journal       = {Reports on Progress in Physics},
  volume        = {88},
  pages         = {026502},
  year          = {2025},
  doi           = {10.1088/1361-6633/ad9756},
  eprint        = {2403.07111},
  archivePrefix = {arXiv},
  primaryClass  = {cond-mat.dis-nn}
}

@article{LB72,
  author  = {Lieb, Elliott H. and Robinson, Derek W.},
  title   = {The finite group velocity of quantum spin systems},
  journal = {Communications in Mathematical Physics},
  volume  = {28},
  number  = {3},
  pages   = {251--257},
  year    = {1972},
  doi     = {10.1007/BF01645779}
}

@book{LSA12,
  author    = {Lewenstein, Maciej and Sanpera, Anna and Ahufinger, Ver{\`o}nica},
  title     = {Ultracold Atoms in Optical Lattices: Simulating Quantum Many-Body Systems},
  publisher = {Oxford University Press},
  address   = {Oxford},
  year      = {2012},
  isbn      = {9780199573127}
}

@article{TL25,
  author        = {Tononi, Andrea and Lewenstein, Maciej},
  title         = {Temporal Bell inequalities in a many-body system},
  journal       = {Quantum Science and Technology},
  volume        = {10},
  pages         = {03LT01},
  year          = {2025},
  doi           = {10.1088/2058-9565/adcbce},
  eprint        = {2409.17290},
  archivePrefix = {arXiv},
  primaryClass  = {quant-ph}
}

@misc{MFL25,
  author        = {M{\"u}ller-Rigat, Guillem and Farina, Donato and Lewenstein, Maciej and Tononi, Andrea},
  title         = {Spatial Leggett--Garg Inequalities},
  year          = {2025},
  eprint        = {2507.03440},
  archivePrefix = {arXiv},
  primaryClass  = {quant-ph}
}

@article{BBKS86,
  author  = {Bardeen, James M. and Bond, J. Richard and Kaiser, Nick and Szalay, Alexander S.},
  title   = {The statistics of peaks of Gaussian random fields},
  journal = {The Astrophysical Journal},
  volume  = {304},
  pages   = {15--61},
  year    = {1986},
  doi     = {10.1086/164143}
}

@book{Dodelson03,
  author    = {Dodelson, Scott},
  title     = {Modern Cosmology},
  publisher = {Academic Press},
  address   = {San Diego, CA},
  year      = {2003},
  isbn      = {9780122191411}
}

@book{Mukhanov05,
  author    = {Mukhanov, Viatcheslav},
  title     = {Physical Foundations of Cosmology},
  publisher = {Cambridge University Press},
  address   = {Cambridge},
  year      = {2005},
  isbn      = {9780521563987}
}

@article{TakahashiNakamura03,
  author        = {Takahashi, Ryuichi and Nakamura, Takashi},
  title         = {Wave Effects in Gravitational Lensing of Gravitational Waves from Chirping Binaries},
  journal       = {The Astrophysical Journal},
  volume        = {595},
  pages         = {1039--1051},
  year          = {2003},
  doi           = {10.1086/377430},
  eprint        = {astro-ph/0305055},
  archivePrefix = {arXiv},
  primaryClass  = {astro-ph}
}

@article{Macquart07,
  author        = {Macquart, J.-P.},
  title         = {Scattering of gravitational radiation: intensity fluctuations},
  journal       = {Astronomy \& Astrophysics},
  volume        = {463},
  number        = {1},
  pages         = {31--49},
  year          = {2007},
  doi           = {10.1051/0004-6361:20065858},
  eprint        = {astro-ph/0402661},
  archivePrefix = {arXiv},
  primaryClass  = {astro-ph}
}

@article{Rothstein13,
  author        = {Rothstein, Ira Z.},
  title         = {Gravitational Anderson Localization},
  journal       = {Physical Review Letters},
  volume        = {110},
  pages         = {011601},
  year          = {2013},
  doi           = {10.1103/PhysRevLett.110.011601},
  eprint        = {1211.7149},
  archivePrefix = {arXiv},
  primaryClass  = {hep-th}
}

@article{PS08,
  author        = {Pikovsky, Arkady S. and Shepelyansky, Dima L.},
  title         = {Destruction of Anderson Localization by a Weak Nonlinearity},
  journal       = {Physical Review Letters},
  volume        = {100},
  pages         = {094101},
  year          = {2008},
  doi           = {10.1103/PhysRevLett.100.094101},
  eprint        = {0708.3315},
  archivePrefix = {arXiv},
  primaryClass  = {cond-mat.dis-nn}
}

@article{FKS09,
  author  = {Flach, Sergej and Krimer, Dmitry O. and Skokos, Charalampos},
  title   = {Universal Spreading of Wave Packets in Disordered Nonlinear Systems},
  journal = {Physical Review Letters},
  volume  = {102},
  pages   = {024101},
  year    = {2009},
  doi     = {10.1103/PhysRevLett.102.024101}
}

@article{OR,
  author =        {Bassi, Angelo and Lochan, Kinjalk and Satin, Seema and
                   Singh, Tejinder P. and Ulbricht, Hendrik},
  journal =       {Reviews of Modern Physics},
  month =         {Apr},
  pages =         {471--527},
  publisher =     {American Physical Society},
  title =         {Models of wave-function collapse, underlying
                   theories, and experimental tests},
  volume =        {85},
  year =          {2013},
  doi =           {10.1103/RevModPhys.85.471},
}

@article{QG_RMP,
  author =        {Bose, Sougato and Fuentes, Ivette and
                   Geraci, Andrew A. and Khan, Saba Mehsar and
                   Qvarfort, Sofia and Rademacher, Markus and
                   Rashid, Muddassar and
                   Toro\ifmmode \check{s}\else \v{s}\fi{}, Marko and
                   Ulbricht, Hendrik and Wanjura, Clara C.},
  journal =       {Rev. Mod. Phys.},
  month =         {Feb},
  pages =         {015003},
  publisher =     {American Physical Society},
  title =         {Massive quantum systems as interfaces of quantum
                   mechanics and gravity},
  volume =        {97},
  year =          {2025},
  doi =           {10.1103/RevModPhys.97.015003},
}

@article{Zurek_decoh,
  author =        {Zurek, Wojciech Hubert},
  journal =       {Review Modern Physics},
  month =         {May},
  pages =         {715--775},
  publisher =     {American Physical Society},
  title =         {Decoherence, einselection, and the quantum origins of
                   the classical},
  volume =        {75},
  year =          {2003},
  doi =           {10.1103/RevModPhys.75.715},
}

@article{Superpos2019,
  author =        {Fein, Yaakov Y and Geyer, Philipp and Zwick, Patrick and
                   Kia{\l}ka, Filip and Pedalino, Sebastian and
                   Mayor, Marcel and Gerlich, Stefan and Arndt, Markus},
  journal =       {Nature Physics},
  number =        {12},
  pages =         {1242--1245},
  publisher =     {Nature Publishing Group},
  title =         {Quantum superposition of molecules beyond 25 {kDa}},
  volume =        {15},
  year =          {2019},
}

@article{Diosi21,
  author =        {Donadi, Sandro and Piscicchia, Kristian and
                   Curceanu, Catalina and Di{\'o}si, Lajos and
                   Laubenstein, Matthias and Bassi, Angelo},
  journal =       {Nature Physics},
  number =        {1},
  pages =         {74--78},
  publisher =     {Nature Publishing Group},
  title =         {Underground test of gravity-related wave function
                   collapse},
  volume =        {17},
  year =          {2021},
  doi =           {10.1038/s41567-020-1008-4},
}

@article{collapse22,
  author =        {Carlesso, Matteo and Donadi, Sandro and
                   Ferialdi, Luca and Paternostro, Mauro and
                   Ulbricht, Hendrik and Bassi, Angelo},
  journal =       {Nature Physics},
  number =        {3},
  pages =         {243--250},
  publisher =     {Nature Publishing Group},
  title =         {Present status and future challenges of
                   non-interferometric tests of collapse models},
  volume =        {18},
  year =          {2022},
  doi =           {10.1038/s41567-021-01489-5},
}

@article{Diosi87,
  author =        {L. Di\'osi},
  journal =       {Physics Letters A},
  number =        {8},
  pages =         {377--381},
  title =         {A universal master equation for the gravitational
                   violation of quantum mechanics},
  volume =        {120},
  year =          {1987},
  doi =           {https://doi.org/10.1016/0375-9601(87)90681-5},
  issn =          {0375-9601},
}

@article{Penrose96,
  author =        {Penrose, Roger},
  journal =       {General Relativity and Gravitation},
  number =        {5},
  pages =         {581--600},
  publisher =     {Springer},
  title =         {On gravity's role in quantum state reduction},
  volume =        {28},
  year =          {1996},
  doi =           {10.1007/BF02105068},
}

@article{SN_review_14,
  author =        {Bahrami, Mohammad and Gro$\beta$ardt, Andr\'e' and
                   Donadi, Sandro and Bassi, Angelo},
  journal =       {New Journal of Physics},
  month =         {nov},
  number =        {11},
  pages =         {115007},
  publisher =     {IOP Publishing},
  title =         {The {S}chr\"odinger--{N}ewton equation and its
                   foundations},
  volume =        {16},
  year =          {2014},
  doi =           {10.1088/1367-2630/16/11/115007},
}

@article{IBB_nonlinear,
  author =        {Iwo Bia{\l}ynicki-Birula and Jerzy Mycielski},
  journal =       {Annals of Physics},
  number =        {1},
  pages =         {62--93},
  title =         {Nonlinear wave mechanics},
  volume =        {100},
  year =          {1976},
  doi =           {https://doi.org/10.1016/0003-4916(76)90057-9},
  issn =          {0003-4916},
}

@article{WeinbergNQM,
  author =        {Steven Weinberg},
  journal =       {Annals of Physics},
  number =        {2},
  pages =         {336--386},
  title =         {Testing quantum mechanics},
  volume =        {194},
  year =          {1989},
  doi =           {https://doi.org/10.1016/0003-4916(89)90276-5},
  issn =          {0003-4916},
}

@article{WeinbergNQM2,
  author =        {Weinberg, Steven},
  journal =       {Physical Review Letters},
  month =         {Jan},
  pages =         {485--488},
  publisher =     {American Physical Society},
  title =         {Precision Tests of Quantum Mechanics},
  volume =        {62},
  year =          {1989},
  doi =           {10.1103/PhysRevLett.62.485},
}

@article{Czachor1998,
  author =        {Czachor, Marek},
  journal =       {Physical Review A},
  month =         {Jun},
  pages =         {4122--4129},
  publisher =     {American Physical Society},
  title =         {Nonlocal-looking equations can make nonlinear quantum
                   dynamics local},
  volume =        {57},
  year =          {1998},
  doi =           {10.1103/PhysRevA.57.4122},
}

@article{CzachorMarcin,
  author =        {Marek Czachor and Marcin Marciniak},
  journal =       {Physics Letters A},
  number =        {6},
  pages =         {353-358},
  title =         {Density matrix interpretation of solutions of
                   {L}ie--{N}ambu equations},
  volume =        {239},
  year =          {1998},
  doi =           {https://doi.org/10.1016/S0375-9601(98)00047-4},
  issn =          {0375-9601},
}

@article{Czachor02,
  author =        {Czachor, Marek and Doebner, H-D},
  journal =       {Physics Letters A},
  number =        {3-4},
  pages =         {139--152},
  publisher =     {Elsevier},
  title =         {Correlation experiments in nonlinear quantum
                   mechanics},
  volume =        {301},
  year =          {2002},
  doi =           {10.1016/S0375-9601(02)00959-3},
}

@article{Kent05,
  author =        {Kent, Adrian},
  journal =       {Phys. Rev. A},
  month =         {Jul},
  pages =         {012108},
  publisher =     {American Physical Society},
  title =         {Nonlinearity without superluminality},
  volume =        {72},
  year =          {2005},
  doi =           {10.1103/PhysRevA.72.012108},
}

@article{Helou17,
  author =        {Bassam Helou and Yanbei Chen},
  journal =       {Journal of Physics: Conference Series},
  number =        {1},
  pages =         {012021},
  publisher =     {IOP Publishing},
  title =         {Extensions of {Born}'s rule to non-linear quantum
                   mechanics, some of which do not imply superluminal
                   communication},
  volume =        {880},
  year =          {2017},
  doi =           {10.1088/1742-6596/880/1/012021},
}

@article{Caban20,
  author =        {Rembieli\ifmmode \acute{n}\else \'{n}\fi{}ski, Jakub and
                   Caban, Pawe\l{}},
  journal =       {Phys. Rev. Res.},
  month =         {Jan},
  pages =         {012027},
  publisher =     {American Physical Society},
  title =         {Nonlinear evolution and signaling},
  volume =        {2},
  year =          {2020},
  doi =           {10.1103/PhysRevResearch.2.012027},
}

@article{Caban21,
  author =        {Rembieli{\'{n}}ski, Jakub and Caban, Pawe{\l{}}},
  journal =       {{Quantum}},
  month =         mar,
  pages =         {420},
  publisher =     {{Verein zur F{\"{o}}rderung des Open Access
                   Publizierens in den Quantenwissenschaften}},
  title =         {Nonlinear extension of the quantum dynamical
                   semigroup},
  volume =        {5},
  year =          {2021},
  doi =           {10.22331/q-2021-03-23-420},
  issn =          {2521-327X},
}

@article{Kaplan2022,
  author =        {Kaplan, David E. and Rajendran, Surjeet},
  journal =       {Phys. Rev. D},
  month =         {Mar},
  pages =         {055002},
  publisher =     {American Physical Society},
  title =         {Causal framework for nonlinear quantum mechanics},
  volume =        {105},
  year =          {2022},
  doi =           {10.1103/PhysRevD.105.055002},
}

@article{Paterek24,
  author =        {Jacek Aleksander Gruca and Ankit Kumar and
                   Ray Ganardi and Paramasivan Arumugam and
                   Karolina Kropielnicka and Tomasz Paterek},
  journal =       {Classical and Quantum Gravity},
  month =         {nov},
  number =        {24},
  pages =         {245014},
  publisher =     {IOP Publishing},
  title =         {Correlations and signaling in the
                   {S}chr\"odinger--{N}ewton model},
  volume =        {41},
  year =          {2024},
  doi =           {10.1088/1361-6382/ad8f8a},
}

@article{NJP2025,
  author =        {Bieli\'nska, Marta Emilia and Eckstein, Micha{\l} and
                   Horodecki, Pawe{\l}},
  journal =       {New Journal of Physics},
  month =         {may},
  number =        {5},
  pages =         {053002},
  publisher =     {IOP Publishing},
  title =         {Superluminal signalling and chaos in nonlinear
                   quantum dynamics},
  volume =        {27},
  year =          {2025},
  doi =           {10.1088/1367-2630/adcae1},
}

@article{Gisin1989,
  author =        {Gisin, Nicolas},
  journal =       {Helvetica Physica Acta},
  number =        {4},
  pages =         {363--371},
  title =         {Stochastic quantum dynamics and relativity},
  volume =        {62},
  year =          {1989},
}

@article{Gisin2001,
  author =        {Simon, Christoph and Bu\v{z}ek, Vladim\'{\i}r and
                   Gisin, Nicolas},
  journal =       {Physical Review Letters},
  month =         {Oct},
  pages =         {170405},
  publisher =     {American Physical Society},
  title =         {No-Signaling Condition and Quantum Dynamics},
  volume =        {87},
  year =          {2001},
  doi =           {10.1103/PhysRevLett.87.170405},
}

@article{Hegerfeldt1,
  author =        {Hegerfeldt, Gerhard C.},
  journal =       {Physical Review D},
  pages =         {3320--3321},
  publisher =     {American Physical Society},
  title =         {Remark on causality and particle localization},
  volume =        {10},
  year =          {1974},
  doi =           {10.1103/PhysRevD.10.3320},
}

@article{Yngvason,
  author =        {Buchholz, Detlev and Yngvason, Jakob},
  journal =       {Physical Review Letters},
  month =         {Aug},
  pages =         {613--616},
  publisher =     {American Physical Society},
  title =         {There are no causality problems for {F}ermi's
                   two-atom system},
  volume =        {73},
  year =          {1994},
  doi =           {10.1103/PhysRevLett.73.613},
}

@article{QIandGR,
  author =        {Peres, Asher and Terno, Daniel R.},
  journal =       {Reviews of Modern Physics},
  month =         {Jan},
  pages =         {93--123},
  publisher =     {American Physical Society},
  title =         {Quantum information and relativity theory},
  volume =        {76},
  year =          {2004},
  doi =           {10.1103/RevModPhys.76.93},
}

@article{Hegerfeldt1985,
  author =        {Hegerfeldt, Gerhard C.},
  journal =       {Physical Review Letters},
  pages =         {2395--2398},
  publisher =     {American Physical Society},
  title =         {Violation of causality in relativistic quantum
                   theory?},
  volume =        {54},
  year =          {1985},
  doi =           {10.1103/PhysRevLett.54.2395},
}

@article{PRA2017,
  author =        {Eckstein, Micha\l{} and Miller, Tomasz},
  journal =       {Physical Review A},
  month =         {Mar},
  pages =         {032106},
  publisher =     {American Physical Society},
  title =         {Causal evolution of wave packets},
  volume =        {95},
  year =          {2017},
  doi =           {10.1103/PhysRevA.95.032106},
}

@article{AHP2017,
  author =        {Eckstein, Micha{\l} and Miller, Tomasz},
  journal =       {Annales Henri Poincar{\'e}},
  pages =         {3049--3096},
  title =         {Causality for nonlocal phenomena},
  volume =        {18},
  year =          {2017},
  doi =           {10.1007/s00023-017-0566-1},
  issn =          {1424-0661},
}

@article{PRA2020,
  author =        {Eckstein, Micha\l{} and Horodecki, Pawe\l{} and
                   Miller, Tomasz and Horodecki, Ryszard},
  journal =       {Physical Review A},
  month =         {Apr},
  pages =         {042128},
  publisher =     {American Physical Society},
  title =         {Operational causality in spacetime},
  volume =        {101},
  year =          {2020},
  doi =           {10.1103/PhysRevA.101.042128},
}

@article{Bell_Nonlocal,
  author =        {Brunner, Nicolas and Cavalcanti, Daniel and
                   Pironio, Stefano and Scarani, Valerio and
                   Wehner, Stephanie},
  journal =       {Reviews of Modern Physics},
  month =         {Apr},
  pages =         {419--478},
  publisher =     {American Physical Society},
  title =         {Bell nonlocality},
  volume =        {86},
  year =          {2014},
  doi =           {10.1103/RevModPhys.86.419},
}

@article{PR_box,
  author =        {Popescu, Sandu},
  journal =       {Nature Physics},
  number =        {4},
  pages =         {264},
  publisher =     {Nature Publishing Group},
  title =         {Nonlocality beyond quantum mechanics},
  volume =        {10},
  year =          {2014},
  doi =           {10.1038/nphys2916},
}

@article{PawelRaviCausality,
  author =        {Horodecki, Pawe{\l} and Ramanathan, Ravishankar},
  journal =       {Nature Communications},
  number =        {1},
  pages =         {1701},
  publisher =     {Nature Publishing Group},
  title =         {The relativistic causality versus no-signaling
                   paradigm for multi-party correlations},
  volume =        {10},
  year =          {2019},
  doi =           {10.1038/s41467-019-09505-2},
}

@article{WSWSG11,
  author =        {Wagner, R. E. and Shields, B. T. and Ware, M. R. and
                   Su, Q. and Grobe, R.},
  journal =       {Physical Review A},
  pages =         {062106},
  publisher =     {American Physical Society},
  title =         {Causality and relativistic localization in
                   one-dimensional {H}amiltonians},
  volume =        {83},
  year =          {2011},
  doi =           {10.1103/PhysRevA.83.062106},
}

@article{Giulini12,
  author =        {Giulini, Domenico and Gro$\beta$ardt, Andr\'e'},
  journal =       {Classical and Quantum Gravity},
  month =         {oct},
  number =        {21},
  pages =         {215010},
  publisher =     {IOP Publishing},
  title =         {The {S}chr\"odinger--{N}ewton equation as a
                   non-relativistic limit of self-gravitating
                   {K}lein--{G}ordon and {D}irac fields},
  volume =        {29},
  year =          {2012},
  doi =           {10.1088/0264-9381/29/21/215010},
}

@book{Wald,
  address =       {Chicago},
  author =        {Wald, Robert M},
  publisher =     {University of Chicago Press},
  title =         {General Relativity},
  year =          {1984},
}

@book{Beem,
  author =        {Beem, J.K. and Ehrlich, P. and Easley, K.},
  publisher =     {CRC Press},
  series =        {{M}onographs and {T}extbooks in {P}ure and {A}pplied
                   {M}athematics},
  title =         {{G}lobal {L}orentzian {G}eometry},
  volume =        {202},
  year =          {1996},
}

@article{MingRev,
  author =        {Minguzzi, Ettore},
  journal =       {Living Reviews in Relativity},
  number =        {1},
  pages =         {1--202},
  publisher =     {Springer},
  title =         {{L}orentzian causality theory},
  volume =        {22},
  year =          {2019},
  doi =           {10.1007/s41114-019-0019-x},
}

@incollection{MS08,
  author =        {Minguzzi, Ettore and S{\'a}nchez, Miguel},
  booktitle =     {Recent Developments in Pseudo-{R}iemannian Geometry,
                   ESI Lectures in Mathematics and Physics},
  editor =        {Alekseevsky, Dmitri V. and Baum, Helga},
  pages =         {299--358},
  publisher =     {European Mathematical Society Publishing House},
  title =         {The causal hierarchy of spacetimes},
  year =          {2008},
}

@book{Ringstrom,
  address =       {Z\"urich},
  author =        {Ringstr{\"o}m, Hans},
  publisher =     {European Mathematical Society},
  series =        {ESI Lectures in Mathematics and Physics},
  title =         {The Cauchy Problem in General Relativity},
  volume =        {6},
  year =          {2009},
  doi =           {10.4171/053},
}

@incollection{Ehlers,
  author =        {Elhers, J and Pirani, FAE and Schild, A},
  booktitle =     {{G}eneral {R}elativity, papers in honour of {J}. {L}.
                   {S}ynge},
  editor =        {O'Raifeartaigh, L},
  note =          {Republished in: \textit{General Relativity and
                   Gravitation} \textbf{44}, 1587 (2012)},
  publisher =     {Clarendon Press, Oxford--London},
  title =         {The geometry of free fall and light propagation},
  year =          {1972},
  doi =           {10.1007/s10714-012-1353-4},
}

@article{PopescuRohrlich94,
  author =        {Popescu, Sandu and Rohrlich, Daniel},
  journal =       {Foundations of Physics},
  number =        {3},
  pages =         {379--385},
  title =         {Quantum nonlocality as an axiom},
  volume =        {24},
  year =          {1994},
  doi =           {10.1007/BF02058098},
  issn =          {1572-9516},
}

@article{NPA08,
  author =        {Navascu{\'e}s, Miguel and Pironio, Stefano and
                   Ac{\'\i}n, Antonio},
  journal =       {New Journal of Physics},
  number =        {7},
  pages =         {073013},
  publisher =     {IOP Publishing},
  title =         {A convergent hierarchy of semidefinite programs
                   characterizing the set of quantum correlations},
  volume =        {10},
  year =          {2008},
  doi =           {10.1088/1367-2630/10/7/073013},
}

@article{AlmostQ,
  author =        {Navascu{\'e}s, Miguel and Guryanova, Yelena and
                   Hoban, Matty J and Ac{\'\i}n, Antonio},
  journal =       {Nature Communications},
  number =        {1},
  pages =         {1--7},
  title =         {Almost quantum correlations},
  volume =        {6},
  year =          {2015},
  doi =           {10.1038/ncomms7288},
}

@article{AQ_hierarchy,
  author =        {Ac{\'\i}n, Antonio and Fritz, Tobias and
                   Leverrier, Anthony and Sainz, Ana Bel{\'e}n},
  journal =       {Communications in Mathematical Physics},
  number =        {2},
  pages =         {533--628},
  publisher =     {Springer},
  title =         {A combinatorial approach to nonlocality and
                   contextuality},
  volume =        {334},
  year =          {2015},
  doi =           {10.1007/s00220-014-2260-1},
}

@article{Miller17a,
  author =        {Tomasz Miller},
  journal =       {Journal of Geometry and Physics},
  pages =         {295--315},
  title =         {Polish spaces of causal curves},
  volume =        {116},
  year =          {2017},
  doi =           {http://dx.doi.org/10.1016/j.geomphys.2017.02.006},
  issn =          {0393-0440},
}

@article{Miller21,
  author =        {Miller, Tomasz},
  journal =       {Advances in Theoretical and Mathematical Physics},
  number =        {3},
  pages =         {743--804},
  title =         {Causal evolution of probability measures and
                   continuity equation},
  volume =        {28},
  year =          {2024},
  doi =           {10.4310/ATMP.241028222555},
}

@article{JGP2021,
  author =        {Tomasz Miller and Micha{\l} Eckstein and
                   Pawe{\l} Horodecki and Ryszard Horodecki},
  journal =       {Journal of Geometry and Physics},
  pages =         {103990},
  title =         {Generally covariant {$N$}-particle dynamics},
  volume =        {160},
  year =          {2021},
  doi =           {10.1016/j.geomphys.2020.103990},
  issn =          {0393-0440},
}

@article{Gerlach1968,
  author =        {Gerlach, Bernd and Gromes, Dieter and
                   Petzold, Joachim and Rosenthal, Peter},
  journal =       {Zeitschrift f{\"{u}}r Physik A Hadrons and nuclei},
  number =        {4},
  pages =         {381--389},
  title =         {{\"{U}}ber kausales {Verhalten} nichtlokaler
                   {Gr{\"{o}}{\ss}en} und {Teilchenstruktur} in der
                   {Feldtheorie}},
  volume =        {208},
  year =          {1968},
  doi =           {10.1007/BF01382700},
  issn =          {0939-7922},
}

@article{Gerlach1969,
  author =        {Gerlach, Bernd and Gromes, Dieter and
                   Petzold, Joachim},
  journal =       {Zeitschrift f{\"{u}}r Physik A Hadrons and nuclei},
  number =        {2},
  pages =         {141--157},
  title =         {Energie und {Kausalit{\"{a}}t}},
  volume =        {221},
  year =          {1969},
  doi =           {10.1007/BF01392139},
  issn =          {0939-7922},
}

@article{Gromes1970,
  author =        {Gromes, Dieter},
  journal =       {Zeitschrift f{\"{u}}r Physik},
  number =        {3},
  pages =         {276--287},
  title =         {On the problem of macrocausality in field theory},
  volume =        {236},
  year =          {1970},
  doi =           {10.1007/BF01394507},
  issn =          {0044-3328},
}

@article{Suhr16,
  author =        {Suhr, Stefan},
  journal =       {M\"unster Journal of Mathematics},
  pages =         {13--47},
  title =         {Theory of optimal transport for {L}orentzian cost
                   functions},
  volume =        {11},
  year =          {2018},
  doi =           {10.17879/87109580432},
}

@incollection{Yngvason2015,
  author =        {Yngvason, Jakob},
  booktitle =     {The Message of Quantum Science: Attempts Towards a
                   Synthesis},
  editor =        {Blanchard, Philippe and Fr{\"o}hlich, J{\"u}rg},
  pages =         {325--348},
  publisher =     {Springer Berlin Heidelberg},
  title =         {Localization and Entanglement in Relativistic Quantum
                   Physics},
  year =          {2015},
  doi =           {10.1007/978-3-662-46422-9_15},
  isbn =          {978-3-662-46422-9},
}

@article{Penrose98,
  author =        {Penrose, Roger},
  journal =       {Philosophical Transactions of the Royal Society of
                   London. Series A: Mathematical, Physical and
                   Engineering Sciences},
  number =        {1743},
  pages =         {1927--1939},
  title =         {Quantum computation, entanglement and state
                   reduction},
  volume =        {356},
  year =          {1998},
  doi =           {10.1098/rsta.1998.0256},
}

@article{Penrose2014,
  author =        {Penrose, Roger},
  journal =       {Foundations of Physics},
  number =        {5},
  pages =         {557--575},
  publisher =     {Springer},
  title =         {On the gravitization of quantum mechanics 1: Quantum
                   state reduction},
  volume =        {44},
  year =          {2014},
  doi =           {10.1007/s10701-013-9770-0},
}

@article{Diosi89,
  author =        {Di\'osi, L.},
  journal =       {Phys. Rev. A},
  month =         {Aug},
  pages =         {1165--1174},
  publisher =     {American Physical Society},
  title =         {Models for universal reduction of macroscopic quantum
                   fluctuations},
  volume =        {40},
  year =          {1989},
  doi =           {10.1103/PhysRevA.40.1165},
}

@article{Moroz98,
  author =        {Irene M Moroz and Roger Penrose and Paul Tod},
  journal =       {Classical and Quantum Gravity},
  month =         {sep},
  number =        {9},
  pages =         {2733},
  publisher =     {},
  title =         {Spherically-symmetric solutions of the
                   {S}chr\"odinger--{N}ewton equations},
  volume =        {15},
  year =          {1998},
  doi =           {10.1088/0264-9381/15/9/019},
}

@article{Bernstein98,
  author =        {Bernstein, David H. and Giladi, Eldar and
                   Jones, Kingsley R. W.},
  journal =       {Modern Physics Letters A},
  number =        {29},
  pages =         {2327-2336},
  title =         {Eigenstates of the gravitational {S}chr\"odinger
                   equation},
  volume =        {13},
  year =          {1998},
  doi =           {10.1142/S0217732398002473},
}

@article{Tod2001,
  author =        {K.P. Tod},
  journal =       {Physics Letters A},
  number =        {4},
  pages =         {173-176},
  title =         {The ground state energy of the
                   {S}chr\"odinger--{N}ewton equation},
  volume =        {280},
  year =          {2001},
  doi =           {https://doi.org/10.1016/S0375-9601(01)00059-7},
  issn =          {0375-9601},
}

@phdthesis{Harrison2001,
  author =        {Harrison, R},
  school =        {{U}niversity of {O}xford},
  title =         {A numerical study of the {S}chr\"odinger--{N}ewton
                   equations},
  year =          {2001},
  url =           {https://ora.ox.ac.uk/objects/uuid:b36a580a-ade8-49a2-ad80-
                  1048a2652b9f/files/m8cfd698e7d5d73ade2cefa24035d5ab4},
}

@article{Harrison2003,
  author =        {R Harrison and I Moroz and K P Tod},
  journal =       {Nonlinearity},
  month =         {nov},
  number =        {1},
  pages =         {101},
  publisher =     {},
  title =         {A numerical study of the {S}chr\"odinger--{N}ewton
                   equations},
  volume =        {16},
  year =          {2002},
  doi =           {10.1088/0951-7715/16/1/307},
}

@phdthesis{Salzman2005,
  author =        {Salzman, Peter Jay},
  school =        {{U}niversity of {C}alifornia, {D}avis},
  title =         {Investigation of the time-dependent
                   {S}chr\"odinger--{N}ewton equation},
  year =          {2005},
  url =           {https://dirac.org/physics/dissertation/dissertation.pdf},
}

@article{Carlip2008,
  author =        {Carlip, S},
  journal =       {Classical and Quantum Gravity},
  month =         {jul},
  number =        {15},
  pages =         {154010},
  title =         {Is quantum gravity necessary?},
  volume =        {25},
  year =          {2008},
  doi =           {10.1088/0264-9381/25/15/154010},
}

@article{Giulini2011,
  author =        {Giulini, Domenico and Gro$\beta$ardt, Andr\'e},
  journal =       {Classical and Quantum Gravity},
  month =         {sep},
  number =        {19},
  pages =         {195026},
  publisher =     {},
  title =         {Gravitationally induced inhibitions of dispersion
                   according to the {S}chr\"odinger--{N}ewton equation},
  volume =        {28},
  year =          {2011},
  doi =           {10.1088/0264-9381/28/19/195026},
}

@article{Giulini2013,
  author =        {Giulini, Domenico and Gro$\beta$ardt, Andr\'e'},
  journal =       {Classical and Quantum Gravity},
  month =         {jul},
  number =        {15},
  pages =         {155018},
  publisher =     {IOP Publishing},
  title =         {Gravitationally induced inhibitions of dispersion
                   according to a modified {S}chr\"odinger--{N}ewton
                   equation for a homogeneous-sphere potential},
  volume =        {30},
  year =          {2013},
  doi =           {10.1088/0264-9381/30/15/155018},
}

@article{Meter2011,
  author =        {van Meter, J R},
  journal =       {Classical and Quantum Gravity},
  month =         {sep},
  number =        {21},
  pages =         {215013},
  publisher =     {},
  title =         {{S}chr\"odinger--{N}ewton `collapse' of the
                   wavefunction},
  volume =        {28},
  year =          {2011},
  doi =           {10.1088/0264-9381/28/21/215013},
}

@article{Manfredi2013,
  author =        {Manfredi, Giovanni and Hervieux, Paul-Antoine and
                   Haas, Fernando},
  journal =       {Classical and Quantum Gravity},
  month =         {mar},
  number =        {7},
  pages =         {075006},
  publisher =     {IOP Publishing},
  title =         {Variational approach to the time-dependent
                   {S}chr\"odinger--{N}ewton equations},
  volume =        {30},
  year =          {2013},
  doi =           {10.1088/0264-9381/30/7/075006},
}

@article{DunajskiPenrose23,
  author =        {Maciej Dunajski and Roger Penrose},
  journal =       {Annals of Physics},
  pages =         {169243},
  title =         {Quantum state reduction, and {N}ewtonian twistor
                   theory},
  volume =        {451},
  year =          {2023},
  doi =           {https://doi.org/10.1016/j.aop.2023.169243},
  issn =          {0003-4916},
}

@article{Schr_Poiss_Vlas_Poiss,
  author =        {{Mocz}, Philip and {Lancaster}, Lachlan and
                   {Fialkov}, Anastasia and {Becerra}, Fernando and
                   {Chavanis}, Pierre-Henri},
  journal =       {Phys. Rev. D},
  month =         apr,
  number =        {8},
  pages =         {083519},
  title =         {{Schr{\"o}dinger-Poisson-Vlasov-Poisson
                   correspondence}},
  volume =        {97},
  year =          {2018},
  doi =           {10.1103/PhysRevD.97.083519},
  eid =           {083519},
}

@article{2024PhRvR_VP_SP,
  author =        {{Cappelli}, Luca and {Tacchino}, Francesco and
                   {Murante}, Giuseppe and {Borgani}, Stefano and
                   {Tavernelli}, Ivano},
  journal =       {Physical Review Research},
  month =         mar,
  number =        {1},
  pages =         {013282},
  title =         {{From Vlasov-Poisson to Schr{\"o}dinger-Poisson: Dark
                   matter simulation with a quantum variational time
                   evolution algorithm}},
  volume =        {6},
  year =          {2024},
  doi =           {10.1103/PhysRevResearch.6.013282},
  eid =           {013282},
}

@article{2025arXiv_fluid_analog_SN,
  author =        {{Coles}, Peter and {Gallagher}, Aoibhinn},
  journal =       {arXiv e-prints},
  month =         jul,
  pages =         {arXiv:2507.08583},
  title =         {{Classical Fluid Analogies for Schr{\"o}dinger-Newton
                   Systems}},
  year =          {2025},
  doi =           {10.48550/arXiv.2507.08583},
  eid =           {arXiv:2507.08583},
}

@article{boson_stars_NatCom16,
  author =        {{Roger}, Thomas and {Maitland}, Calum and
                   {Wilson}, Kali and {Westerberg}, Niclas and
                   {Vocke}, David and {Wright}, Ewan M. and
                   {Faccio}, Daniele},
  journal =       {Nature Communications},
  month =         nov,
  pages =         {13492},
  title =         {{Optical analogues of the Newton-Schr{\"o}dinger
                   equation and boson star evolution}},
  volume =        {7},
  year =          {2016},
  doi =           {10.1038/ncomms13492},
  eid =           {13492},
}

@article{Finster98,
  author =        {Finster, Felix},
  journal =       {Journal of Mathematical Physics},
  month =         {12},
  number =        {12},
  pages =         {6276-6290},
  title =         {Local $U(2,2)$ symmetry in relativistic quantum
                   mechanics},
  volume =        {39},
  year =          {1998},
  doi =           {10.1063/1.532638},
  issn =          {0022-2488},
}

@article{Finster99,
  author =        {Finster, Felix and Smoller, Joel and Yau, Shing-Tung},
  journal =       {Physical Review D},
  month =         {Apr},
  pages =         {104020},
  publisher =     {American Physical Society},
  title =         {Particlelike solutions of the {E}instein--{D}irac
                   equations},
  volume =        {59},
  year =          {1999},
  doi =           {10.1103/PhysRevD.59.104020},
}

@book{Thaller,
  author =        {Thaller, Bernd},
  publisher =     {Springer-Verlag Berlin},
  series =        {Theoretical and Mathematical Physics},
  title =         {The {Dirac} {E}quation},
  volume =        {31},
  year =          {1992},
  doi =           {10.1007/978-3-662-02753-0},
}

@article{EinsteinDirac25,
  author =        {{Zhao}, Peng and {Wu}, Xiaoning},
  journal =       {arXiv e-prints},
  month =         sep,
  pages =         {arXiv:2509.04167},
  title =         {{On the local existence for the characteristic
                   initial value problem for the Einstein-Dirac system}},
  year =          {2025},
  doi =           {10.48550/arXiv.2509.04167},
  eid =           {arXiv:2509.04167},
}

@book{DiracCurvedBook,
  author =        {Collas, Peter and Klein, David},
  publisher =     {Springer},
  title =         {The Dirac Equation in Curved Spacetime: A Guide for
                   Calculations},
  year =          {2019},
}

@book{William_H_Press2007-rp,
  address =       {Cambridge},
  author =        {{William H. Press} and Teukolsky, Saul A and
                   Vetterling, William T and Flannery, Brian P},
  edition =       {3},
  month =         sep,
  publisher =     {Cambridge University Press},
  title =         {Numerical recipes 3rd edition},
  year =          {2007},
}

@article{Muruganandam2009,
  author =        {Paulsamy Muruganandam and Sadhan K. Adhikari},
  journal =       {Computer Physics Communications},
  number =        {10},
  pages =         {1888--1912},
  title =         {Fortran programs for the time-dependent
                   {Gross}--{Pitaevskii} equation in a fully anisotropic
                   trap},
  volume =        {180},
  year =          {2009},
  doi =           {10.1016/j.cpc.2009.04.015},
}

@article{Kumar2015Dipolar,
  author =        {R. Kishor Kumar and Luis E. Young-S. and
                   Du{\v{s}}an Vudragovi{\'c} and Antun Bala{\v{z}} and
                   Paulsamy Muruganandam and S. K. Adhikari},
  journal =       {Computer Physics Communications},
  pages =         {117--128},
  title =         {Fortran and {C} programs for the time-dependent
                   dipolar {Gross}--{Pitaevskii} equation in an
                   anisotropic trap},
  volume =        {195},
  year =          {2015},
  doi =           {10.1016/j.cpc.2015.03.024},
}

@article{BaoDu2004,
  author =        {Weizhu Bao and Qiang Du},
  journal =       {SIAM Journal on Scientific Computing},
  number =        {5},
  pages =         {1674--1697},
  title =         {Computing the ground state solution of
                   {Bose}--{Einstein} condensates by a normalized
                   gradient flow},
  volume =        {25},
  year =          {2004},
  doi =           {10.1137/S1064827503422956},
}

@article{AntoineBaoBesse2013,
  author =        {Xavier Antoine and Weizhu Bao and Christophe Besse},
  journal =       {Computer Physics Communications},
  number =        {12},
  pages =         {2621--2633},
  title =         {Computational methods for the dynamics of the
                   nonlinear {Schr{\"o}dinger}/{Gross--Pitaevskii}
                   equations},
  volume =        {184},
  year =          {2013},
  doi =           {10.1016/j.cpc.2013.07.012},
}

@article{BaoJiangTangZhang2015,
  author =        {Weizhu Bao and Shidong Jiang and Qinglin Tang and
                   Yong Zhang},
  journal =       {Journal of Computational Physics},
  pages =         {72--89},
  title =         {Computing the ground state and dynamics of the
                   nonlinear {Schr{\"o}dinger} equation with nonlocal
                   interactions via the nonuniform {FFT}},
  volume =        {296},
  year =          {2015},
  doi =           {10.1016/j.jcp.2015.04.045},
}

@article{BaoTangZhang2016Dipolar,
  author =        {Weizhu Bao and Qinglin Tang and Yong Zhang},
  journal =       {Communications in Computational Physics},
  number =        {5},
  pages =         {1141--1166},
  title =         {Accurate and efficient numerical methods for
                   computing ground states and dynamics of dipolar
                   {Bose}--{Einstein} condensates via the nonuniform
                   {FFT}},
  volume =        {19},
  year =          {2016},
  doi =           {10.4208/cicp.051014.300715a},
}

@article{AntoineDuboscq2014GPELab,
  author =        {Xavier Antoine and Romain Duboscq},
  journal =       {Computer Physics Communications},
  number =        {11},
  pages =         {2969--2991},
  title =         {{G}{P}{E}{Lab}, a {Matlab} toolbox to solve
                   {Gross--Pitaevskii} equations {I}: Computation of
                   stationary solutions},
  volume =        {185},
  year =          {2014},
  doi =           {10.1016/j.cpc.2014.06.026},
}

@article{boucomas2025quanticstensortrainsolving,
  author =        {{Bou-Comas}, Aleix and {P{\l}odzie{\'n}}, Marcin and
                   {Tagliacozzo}, Luca and
                   {Jos{\'e} Garc{\'\i}a-Ripoll}, Juan},
  journal =       {arXiv e-prints},
  month =         jul,
  pages =         {arXiv:2507.03134},
  title =         {{Quantics Tensor Train for solving Gross-Pitaevskii
                   equation}},
  year =          {2025},
  doi =           {10.48550/arXiv.2507.03134},
  eid =           {arXiv:2507.03134},
}

@article{Penrose03,
  author =        {Marshall, William and Simon, Christoph and
                   Penrose, Roger and Bouwmeester, Dik},
  journal =       {Physical Review Letters},
  month =         {Sep},
  pages =         {130401},
  publisher =     {American Physical Society},
  title =         {Towards Quantum Superpositions of a Mirror},
  volume =        {91},
  year =          {2003},
  doi =           {10.1103/PhysRevLett.91.130401},
}

@article{LimitsQM,
  author =        {Arndt, Markus and Hornberger, Klaus},
  journal =       {Nature Physics},
  number =        {4},
  pages =         {271--277},
  publisher =     {Nature Publishing Group},
  title =         {Testing the limits of quantum mechanical
                   superpositions},
  volume =        {10},
  year =          {2014},
  doi =           {10.1038/nphys2863},
}

@article{Howl2019,
  author =        {Howl, Richard and Penrose, Roger and Fuentes, Ivette},
  journal =       {New Journal of Physics},
  month =         {apr},
  number =        {4},
  pages =         {043047},
  publisher =     {IOP Publishing},
  title =         {Exploring the unification of quantum theory and
                   general relativity with a {B}ose–{E}instein
                   condensate},
  volume =        {21},
  year =          {2019},
  doi =           {10.1088/1367-2630/ab104a},
}

@article{SN2022,
  author =        {Sahoo, Sourav Kesharee and Dash, Ashutosh and
                   Vathsan, Radhika and Qureshi, Tabish},
  journal =       {Phys. Rev. A},
  month =         {Jul},
  pages =         {012215},
  publisher =     {American Physical Society},
  title =         {Testing gravitational self-interaction via
                   matter-wave interferometry},
  volume =        {106},
  year =          {2022},
  doi =           {10.1103/PhysRevA.106.012215},
}

@article{Mielnik01,
  author =        {Bogdan Mielnik},
  journal =       {Physics Letters A},
  number =        {1},
  pages =         {1-8},
  title =         {Nonlinear quantum mechanics: a conflict with the
                   {P}tolomean structure?},
  volume =        {289},
  year =          {2001},
  doi =           {10.1016/S0375-9601(01)00583-7},
  issn =          {0375-9601},
}

@article{Diosi25,
  author =        {Diósi, Lajos},
  journal =       {Journal of Physics: Conference Series},
  month =         {jun},
  number =        {1},
  pages =         {012009},
  publisher =     {IOP Publishing},
  title =         {Causality violation of {S}chr\"odinger--{N}ewton
                   equation: direct test on the horizon?},
  volume =        {3017},
  year =          {2025},
  doi =           {10.1088/1742-6596/3017/1/012009},
}

@article{PhysRevLett.84.5687,
  author =        {O'Dell, D. and Giovanazzi, S. and Kurizki, G. and
                   Akulin, V. M.},
  journal =       {Phys. Rev. Lett.},
  month =         {Jun},
  pages =         {5687--5690},
  publisher =     {American Physical Society},
  title =         {Bose--{Einstein} {Condensates} with $1/\mathit{r}$
                   {Interatomic} {Attraction}: {Electromagnetically}
                   {Induced} ``Gravity''},
  volume =        {84},
  year =          {2000},
  doi =           {10.1103/PhysRevLett.84.5687},
}

@article{PhysRevA.63.031603,
  author =        {Giovanazzi, S. and O'Dell, D. and Kurizki, G.},
  journal =       {Phys. Rev. A},
  month =         {Feb},
  pages =         {031603},
  publisher =     {American Physical Society},
  title =         {Self-binding transition in {Bose} condensates with
                   laser-induced ``gravitation''},
  volume =        {63},
  year =          {2001},
  doi =           {10.1103/PhysRevA.63.031603},
}

@article{Giovanazzi_2001,
  author =        {S. Giovanazzi and G. Kurizki and I. E. Mazets and
                   S. Stringari},
  journal =       {Europhysics Letters},
  month =         {oct},
  number =        {1},
  pages =         {1},
  publisher =     {},
  title =         {Collective excitations of a "gravitationally"
                   self-bound {Bose} gas},
  volume =        {56},
  year =          {2001},
  doi =           {10.1209/epl/i2001-00478-8},
}

@article{PhysRevA.65.053616,
  author =        {Ghosh, Tarun Kanti},
  journal =       {Phys. Rev. A},
  month =         {May},
  pages =         {053616},
  publisher =     {American Physical Society},
  title =         {Collective excitation frequencies and vortices of a
                   {Bose}--{Einstein} condensed state with gravitylike
                   interatomic attraction},
  volume =        {65},
  year =          {2002},
  doi =           {10.1103/PhysRevA.65.053616},
}

@article{PhysRevA.76.053604,
  author =        {Papadopoulos, I. and Wagner, P. and Wunner, G. and
                   Main, J.},
  journal =       {Phys. Rev. A},
  month =         {Nov},
  pages =         {053604},
  publisher =     {American Physical Society},
  title =         {Bose--{Einstein} condensates with attractive $1/r$
                   interaction: The case of self-trapping},
  volume =        {76},
  year =          {2007},
  doi =           {10.1103/PhysRevA.76.053604},
}

@article{PhysRevA.78.013615,
  author =        {Cartarius, Holger and
                     Fab\ifmmode \check{c}\else \v{c}\fi{}i\ifmmode
  \check{c}\else \v{c}\fi{}, Toma\ifmmode \check{z}\else \v{z}\fi{} and Main,
  J\"org and Wunner, G\"unter},
  journal =       {Phys. Rev. A},
  month =         {Jul},
  pages =         {013615},
  publisher =     {American Physical Society},
  title =         {Dynamics and stability of {Bose}-{Einstein}
                   condensates with attractive $1/r$ interaction},
  volume =        {78},
  year =          {2008},
  doi =           {10.1103/PhysRevA.78.013615},
}

@article{PhysRevA.82.023615,
  author =        {Lushnikov, Pavel M.},
  journal =       {Phys. Rev. A},
  month =         {Aug},
  pages =         {023615},
  publisher =     {American Physical Society},
  title =         {Collapse and stable self-trapping for
                   {Bose}--{Einstein} condensates with $1/{r}^{b}$-type
                   attractive interatomic interaction potential},
  volume =        {82},
  year =          {2010},
  doi =           {10.1103/PhysRevA.82.023615},
}

@article{Li2014,
  author =        {Li, Jinbin and Qiao, Yaxin},
  journal =       {Journal of Low Temperature Physics},
  month =         jul,
  number =        {3–4},
  pages =         {165–177},
  publisher =     {Springer Science and Business Media LLC},
  title =         {Properties of {Two}-{Component} {Bose}–{Einstein}
                   {Condensates} with {Monopolar} {Interaction}},
  volume =        {177},
  year =          {2014},
  doi =           {10.1007/s10909-014-1207-4},
  issn =          {1573-7357},
}

@article{PENG20167,
  author =        {Xinxin Peng and Jinbin Li},
  journal =       {Physica B: Condensed Matter},
  pages =         {7-12},
  title =         {Mixtures of two-component {B}{E}{C}s with long-range
                   monopolar interaction},
  volume =        {484},
  year =          {2016},
  doi =           {https://doi.org/10.1016/j.physb.2015.12.017},
  issn =          {0921-4526},
}

@article{10.1063/10.0000130,
  author =        {Tamaddonpur, Moulud and Yavari, Heshmatollah and
                   Saeidi, Zahra},
  journal =       {Low Temperature Physics},
  month =         {11},
  number =        {11},
  pages =         {1187-1192},
  title =         {Effect of long-range 1/r interaction on thermal and
                   quantum depletion of a dipolar quasi-two-dimensional
                   {Bose} gas},
  volume =        {45},
  year =          {2019},
  doi =           {10.1063/10.0000130},
  issn =          {1063-777X},
}

@article{PhysRevE.111.064207,
  author =        {Zhao, Zibin and Li, Guilong and Luo, Huanbo and
                   Liu, Bin and Chen, Gui-hua and Malomed, Boris A. and
                   Li, Yongyao},
  journal =       {Phys. Rev. E},
  month =         {Jun},
  pages =         {064207},
  publisher =     {American Physical Society},
  title =         {Tightly self-trapped modes and vortices in
                   three-dimensional bosonic condensates with
                   electromagnetically induced gravity},
  volume =        {111},
  year =          {2025},
  doi =           {10.1103/PhysRevE.111.064207},
}

@article{liu2025detectingcollectiveexcitationsselfgravitating,
  author =        {{Liu}, Ning and {Cheng}, Guodong},
  journal =       {arXiv e-prints},
  month =         jun,
  pages =         {arXiv:2506.18593},
  title =         {{Detecting Collective Excitations in Self-Gravitating
                   Bose-Einstein Condensates via Faraday Waves}},
  year =          {2025},
  doi =           {10.48550/arXiv.2506.18593},
  eid =           {arXiv:2506.18593},
}

\end{document}